\documentclass[paper=a4, 12pt, twoside]{article} 
\usepackage{lastpage}
\newcommand{\pyear}{2026}  
\newcommand{\vol}{24}  
\newcommand{\no}{2} 
\newcommand{\stpage}{1}  

\newcommand{\enpage}{\pageref{LastPage}} 
\usepackage[top=2.54cm, bottom=2.54cm, left=2.54cm, right=2.54cm]{geometry} 
\usepackage{times} 
\usepackage{lmodern} 
\usepackage{hanging} 
\usepackage{graphicx} 
\usepackage[english]{babel}  
\usepackage{natbib}
\usepackage{fancyhdr} 
\usepackage{bbm,amsthm}
\newtheorem{lem}{Lemma}
\newtheorem{rem}{Remark}

\fancypagestyle{frontpagefooter}{%
  \fancyhead{} 
  \fancyfoot[L]{\footnotesize 
    Corresponding Authors: Garrett Frady and Shariq Mohammed \\
    Emails: garrett.frady@uconn.edu; shariqm@bu.edu
    }
  \fancyfoot[C]{\footnotesize }	  
  \fancyfoot[R]{\footnotesize } 
}
\usepackage{amsthm}

\newtheoremstyle{sa}
{3mm} 
{2mm} 
{} 
{} 
{\bfseries} 
{:} 
{.5em} 
{} 

\theoremstyle{sa}

\usepackage{amsmath,amssymb}
\usepackage{hyperref} 
\usepackage{lscape}
\usepackage{arydshln}
\usepackage{booktabs}
\usepackage{verbatim}
\usepackage{enumitem}
\usepackage{bm}
\DeclareMathAlphabet{\mathpzc}{OT1}{pzc}{m}{it}
\usepackage{wasysym}
\usepackage{esint}
\usepackage{xcolor}
\usepackage{tikz}
\usetikzlibrary{shapes,backgrounds,arrows,chains,matrix,positioning,scopes}
\usepackage{caption}
\usepackage{titlesec}

\titleformat{\section}[block]{\bf \raggedright}
{\makebox[1.27cm][l]{\thesection.}}{0in}{} 
\titlespacing*{\section}{0pt}{3mm}{0pt}

\titleformat{\subsection}[block]{\bf}
{\makebox[1.27cm][l]{\thesubsection.}}{0in}{}  
\titlespacing*{\subsection}{0pt}{3mm}{0pt}

\titleformat{\subsubsection}[block]{\bf}
{\makebox[1.27cm][l]{\thesubsubsection.}}{0in}{}  
\titlespacing*{\subsubsection}{0pt}{3mm}{0pt}

\usepackage{indentfirst}  
\usepackage[labelfont=bf,textfont=bf]{caption}  
\begin{document}
\thispagestyle{empty}

\thispagestyle{frontpagefooter}
 
 \begin{minipage}[t]{0.65\textwidth}
 	\vspace{0pt}
 	\begin{flushleft}
 		{Statistics and Applications} {\{ISSN 2454-7395 (online)\}} \\
 		Volume \vol, No. \no, \pyear \ (New Series), pp \stpage--\enpage\\
 		\href{http://www.ssca.org.in/journal.html}{{{http://www.ssca.org.in/journal}}}
 	\end{flushleft}
 \end{minipage}
 \begin{minipage}[t]{0.3\textwidth}
 	\vspace{-2pt}
 	\begin{flushright}
 		\includegraphics[width=0.5\textwidth]{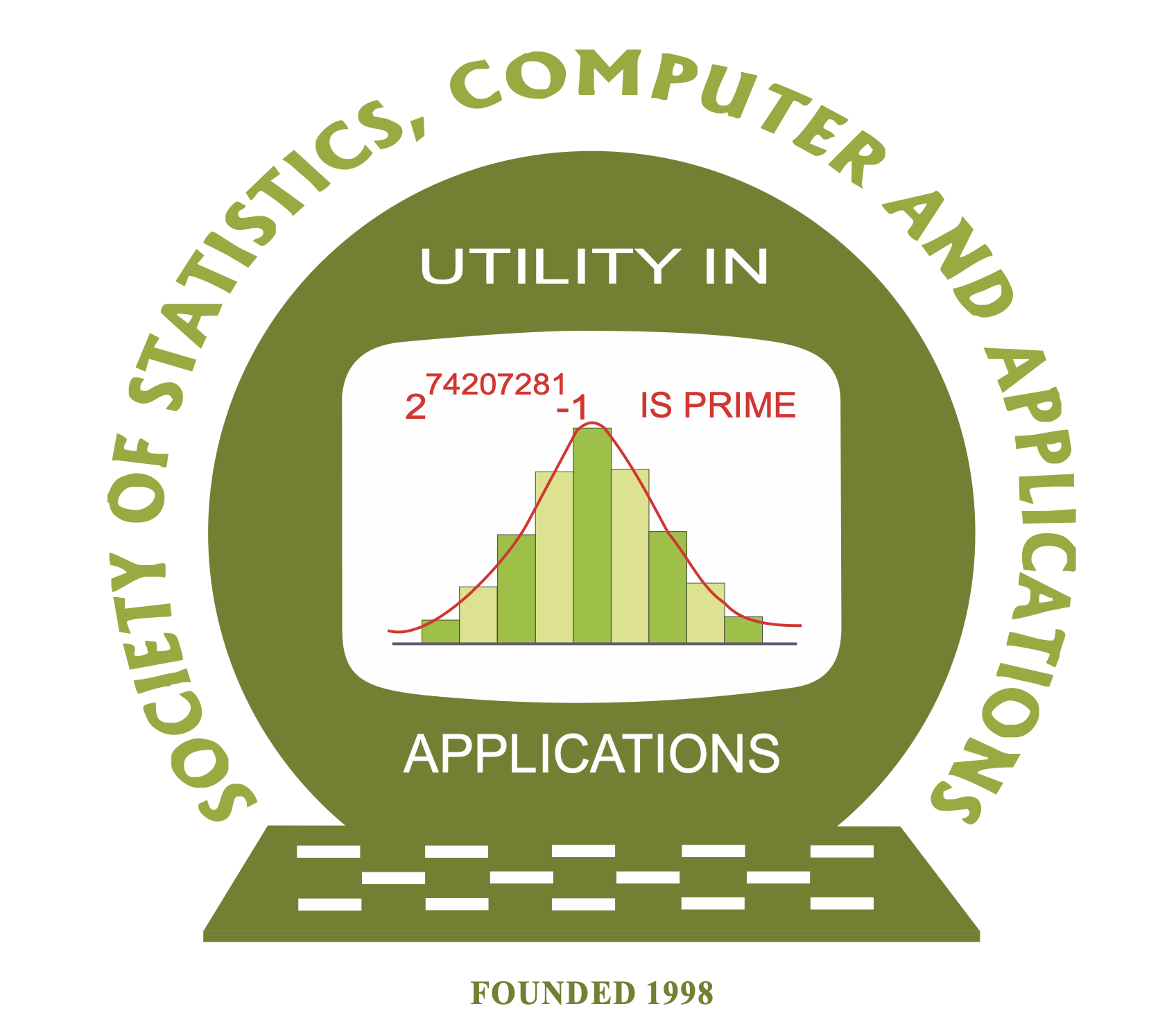}
 	\end{flushright}
 \end{minipage}

%
%

\begin{center}
{\fontsize{16}{24} \selectfont 
{\bf Bayesian Feature Extraction using Gaussian and Diffused-gamma Priors for High Dimensional Spatio-Temporal Data} 
 } \\[6mm]
\normalsize
{\bf Garrett Frady$^{1}$, Dipak K. Dey$^{1}$, and Shariq Mohammed$^{2,3}$} \\ 
{\it $^{1}$Department of Statistics, University of Connecticut, Storrs, Connecticut, U.S.A. \\
$^{2}$Department of Biostatistics, Boston University, Boston, Massachusetts, U.S.A. \\
$^{3}$Rafik B.Hariri Institute for Computing, Computational Science and Engineering, \\ Boston University, Boston, Massachusetts, U.S.A.
}

Received: dd Month yyyy; Revised: dd Month yyyy; Accepted: dd Month yyyy 
\end{center}
\vspace{-6mm}

\setlength\parindent{1.27cm}

\noindent\rule{\linewidth}{0.8pt}\vspace{-3mm}

\noindent{\bf Abstract}\vspace{-1mm}

High-dimensional data with sparse structure and spatio-temporal dependence arise in many scientific domains. We develop a Bayesian feature-extraction framework for spatio-temporal settings that employs Gaussian and Diffused-gamma priors to induce structured sparsity. The modeling framework specifies a general likelihood via Bregman divergence, enabling compatibility with a range of loss functions and measurement models. Posterior computation is carried out via Markov chain Monte Carlo (MCMC), and we introduce a two-stage feature-extraction procedure based on posterior samples to stabilize selection across space and time. We illustrate the method with a multi-subject electroencephalography (EEG) case study examining the relationship between chronic alcohol exposure and activity in different brain regions. We first fit binary classification models at each time point, then use false discovery rate-controlled screening and subsequent clustering in a two‑stage feature‑extraction pipeline to identify active brain regions. The analysis demonstrates that our proposed priors improve recovery of sparse features and enhance interpretability in the presence of spatio-temporal dependence. The framework is broadly applicable to high-dimensional, structured problems where accurate feature selection and inference are required. The code to implement the model is publicly available via GitHub.


\vspace{-5mm} 

\noindent\rule{\linewidth}{0.8pt} 
\setcounter{equation}{0}

\section{Introduction}

\label{s:intro}

Due to vast improvements in technology and data collection methods, performing statistical analyses commonly involves modeling massive amounts of data containing a significantly larger number of parameters than outcomes. In such cases, classical statistical models are inadequate and some form of regularization, dimension reduction, or variable selection is required. Additionally, data often includes matrix covariates with correlations extending across multiple dimensions. For example, predicting natural disasters and discovering global warming trends require analyzing three-dimensional data capturing dynamic spatial patterns changing over time. Beyond the evolution of environmental phenomena, spatio-temporal data analysis can provide valuable insights into distribution and determinants of public health outcomes. Measurements from neuroimaging procedures, such as electroencephalography (EEG) and functional magnetic resonance imaging (fMRI), may be collected in a spatio-temporal pattern to examine statistical dependencies between electrical signals at different regions of the brain. Understanding how locations of the brain communicate plays a major role in detecting early onset of mental-related illnesses, including dementia, opioid addiction, alcoholism, and attention deficit hyperactivity disorder (ADHD). Although spatio-temporal data arises across a wide variety of scientific domains, this paper develops the proposed framework exclusively using EEG data as the motivating example. The minimum data structure required for the framework  to apply is a three-dimensional array of measurements over $n$ subjects,  $L$ spatial locations, and $\tau$ time points, paired with a subject-level  response variable, where $n \ll L \times \tau$.

Designing a supervised predictive model of subject-level behavior using neuroimaging data comes with its difficulties. Brain imaging techniques generally produce tens of thousands to hundreds of thousands of measurements for each subject, acting as covariates, and only consist of tens to hundreds of subjects because of the invasive nature of brain imaging procedures. Different natural biological behavior results in high subject-to-subject variability. The spatio-temporal structure yields two-dimensional predictors and correlation exists across multiple dimensions. Such a correlation structure may not be fully captured in a computationally efficient manner through traditional statistical or machine learning methods. The sheer size of the data will likely be too large for the memory available on a single computer node. Conducting analysis in an efficient manner requires efficient and scalable strategies.

We present a case study utilizing EEG data as our motivating example. The primary goal of this paper is to identify active brain regions that are associated with cognitive or behavioral states implying potential chronic exposure to alcohol. Then, classify subjects as at risk, or not, of alcoholism according to their EEG measurements. There have been approaches in literature that address one or both of these objectives. We conduct a thorough literature review to examine prior methodologies address one or both of these objectives, and highlight new contributions in our approach. We propose a computationally efficient Bayesian-feature extraction method utilizing the Gaussian and Diffused-gamma (GD) prior \citep{goh2018bayesian} by constructing models locally at each time point. Covariates in each model are electrical signals observed at locations. Estimation is carried out using the Markov chain Monte Carlo (MCMC) algorithm and the Metropolis-Hastings algorithm within Gibbs sampler. MCMC samples obtained from the GD model are used to develop a two-stage feature extraction process to make final estimations and select active locations, and a subject-level weighted prediction process to accurately classify subjects. 

A substantial body of research has been allocated to developing accurate and implementable methods to handle data with the complex characteristics previously described. The most common approach to reduce the dimension of a problem is vectorizing the matrix covariates. Since the number of covariates, $p$, in neuroimaging data is enormous, vectorizing the covariates will not bypass the high dimensional issue. To address this and the difficulties described in Section A.1 of the Annexure, \cite{hu2015local} developed a local-aggregate modeling approach applied to EEG data. By constructing models at each location, the dimension of the predictor vectors is equivalent to the number of time points, tending to the order of hundreds up to thousands. Thus, the local modeling approach requires additional effort to reduce the dimension of the supervised learning problem. 

\subsection{Sparsity, Feature Extraction, Penalized Loss Function Estimator}
\label{sec:sparsefeatext}

Even with large $p$, in several real-world applications, only a small portion of the predictors may have a significant impact on the response (leading to sparsity within high dimensional settings). For example, in neuroimaging research, only a few of the brain locations may be significantly correlated to predicting early onset of a mental-related illness. Dimension reduction is a matter of removing irrelevant predictors from the model by forcing corresponding coefficients to zero. To this end, multiple studies \citep{tibshirani1996regression, fan2001variable, zou2006adaptive} have employed the penalized loss function (PL) estimator to achieve simultaneous estimation and feature extraction. The PL estimator is defined as 
\begin{gather}
    \hat{\boldsymbol{\beta}}_{PL} = \underset{\boldsymbol{\beta}}{\text{argmin}} [L(\boldsymbol{y}, \boldsymbol{h(X\beta)}) + Pe(\boldsymbol{\beta} | \lambda)], \label{plest}
\end{gather}
where $L(\cdot, \cdot)$ denotes a loss function, $Pe(\cdot | \lambda)$ is a penalty function, and $\lambda \geq 0$ is a regularization parameter intended to control the degree of penalization. An abundance of research has been devoted to developing sparsity inducing penalty functions. First formally developed by \cite{akaike1974new} and further investigated by \cite{schwarz1978estimating}, the $\ell_0$-norm penalty, $\ell_0(\boldsymbol{\beta} | \lambda) = \lambda \sum_{j = 1}^p \mathbbm{1}_{(\beta_j \neq 0)}$, where $\mathbbm{1}_{(\cdot)}$ is an indicator function, was introduced to induce sparsity into $\hat{\boldsymbol{\beta}}_{PL}$. Quantifying $\hat{\boldsymbol{\beta}}_{PL}$ remains computationally cumbersome caused by the discontinuous nature of indicator functions. Many methods followed in attempt to improve previous methods (see examples in Section A.1 of the Annexure). 

Although PL point estimators are attractive, \cite{casella2010penalized} pointed out that frequentist LASSO methods are unable to produce valid standard errors, rendering them unable to accurately quantify uncertainty in the estimators. In turn, Bayesian techniques became popular in sparse, high dimensional literature. An advantage of our local Bayesian modeling approach over the frequentist local modeling approach in \cite{hu2015local} is that we are able to quantify the uncertainty in the estimators through the posterior distribution of $\boldsymbol{\beta}$. A potential limitation of Bayesian approaches is that for large values of $p$, the MCMC sampling algorithm proves to be computationally inefficient and ineffective \citep{rovckova2014emvs}. To overcome this deficiency in the analysis of EEG data, \cite{mohammed2019bayesian} proposed a Bayesian local-aggregate modeling strategy by constructing continuous spike-and-slab models at each time point. Building models locally at time points instead of locations reduces the dimension further, since EEG data has high temporal resolution but low spatial resolution. An advantage of our GD model over spike-and-slab models is that we have less hyperparameters to tune and do not require extensive computational resources or sensitivity analysis to find optimal values.

\underline{Novelty of the proposed methods}. Existing approaches, such as \cite{goh2018bayesian}, rely on a component-wise iterated conditional modes (ICM) algorithm \citep{besag1986statistical} within MCMC and impose sparsity through a hard threshold determined by a regularization parameter, $\lambda$. Our method addresses these limitations by introducing a fully Bayesian framework that avoids ad hoc thresholding and manual tuning. Specifically, we (i) place a non-informative gamma prior on $\lambda$ rather than performing a grid search as in \cite{goh2018bayesian}, reducing computational burden and improving robustness; and (ii) develop a novel two-stage feature extraction pipeline based on posterior samples, combining FDR-controlled screening with spatial-temporal clustering and weighted subject-level prediction. This integration of structured priors, adaptive sparsity induction, and posterior-driven feature selection represents a significant advancement over existing methods, enabling more reliable recovery of sparse features in high-dimensional spatio-temporal settings.

The outline for the remaining sections in this paper is as follows. In Section \ref{sec:eeg}, we describe the EEG data used in our case study, the procedure and characteristics of its measurements. In Section \ref{sec:lbm}, we explain our local Bayesian modeling framework. We define the GD prior and its relationship with a continuous-and-differentiable $\ell_0$-norm approximation, discuss the role of the class of Bregman divergence measures in our approach, present a computationally feasible MCMC algorithm, and give background theory for posterior computation under a specified model with the GD prior. Additionally, we describe our two-stage feature extraction process and subject-level weighted prediction process. In Section \ref{sec:numstud}, we conduct a simulation study and an EEG case study. We compare the performance of our GD method with existing methods in the literature. Finally, we provide concluding remarks and potential future developments in Section \ref{sec:conclusion}. Additional details and proofs are provided in the Annexure. 

\section{Electroencephalography Data}
\label{sec:eeg}

EEG is a popular non-invasive brain-imaging technique administered as a way to record electrical brain signals over a specified interval of time via electrodes attached to different locations of a subject's scalp. Consequently, researchers commonly conduct EEG analysis to examine the functional connectivity between different parts of the brain in response to a stimulus. From a statistical viewpoint, detecting associated brain regions in such a manner may be viewed as a feature extraction process. The non-invasive aspect of EEG leads to a low signal-to-noise ratio, because the subjects' skull, connective tissue, and body fluid between the brain and EEG electrodes significantly debilitates the electrical signal \citep{tam2019human}. 

In our case study, we use open access EEG data obtained from \url{https://archive.ics.uci.edu/ml/datasets/eeg+database}. A clean version of this data set is available in an \texttt{R} package \texttt{SpikeSlabEEG} on GitHub \citep{mohammed2019bayesian}. The data surfaces from a large study of 120 trials for 122 subjects to investigate EEG correlates of genetic predisposition to alcoholism. Measurements were taken from 64 electrodes placed on the subjects’ scalp by sampling at 256 Hz for 1 second (3.9-ms epochs). Photos from \cite{snodgrass1980standardized} were shown to each subject as a stimulus. In our analysis, we consider 57 of the 64 locations as they are the only locations with coordinates associated to them. Hence, the predictor tensor, $\chi$, has dimensions $122 \times 57 \times 256$. As EEG captures alterations in voltage resulting from the movement of electrically charged ions across neuronal cell membranes, the measurements are continuous. In order to visualize the spatio-temporal behavior of the data, we present event related potential (ERP) plots and scalp maps. For more information on these visualization techniques, see Section A.4 of the Annexure. 

\begin{figure}[!t]
    \centering
    \includegraphics[scale=0.8, trim=0.15cm 0.325cm 0.2cm 0.175cm, clip]{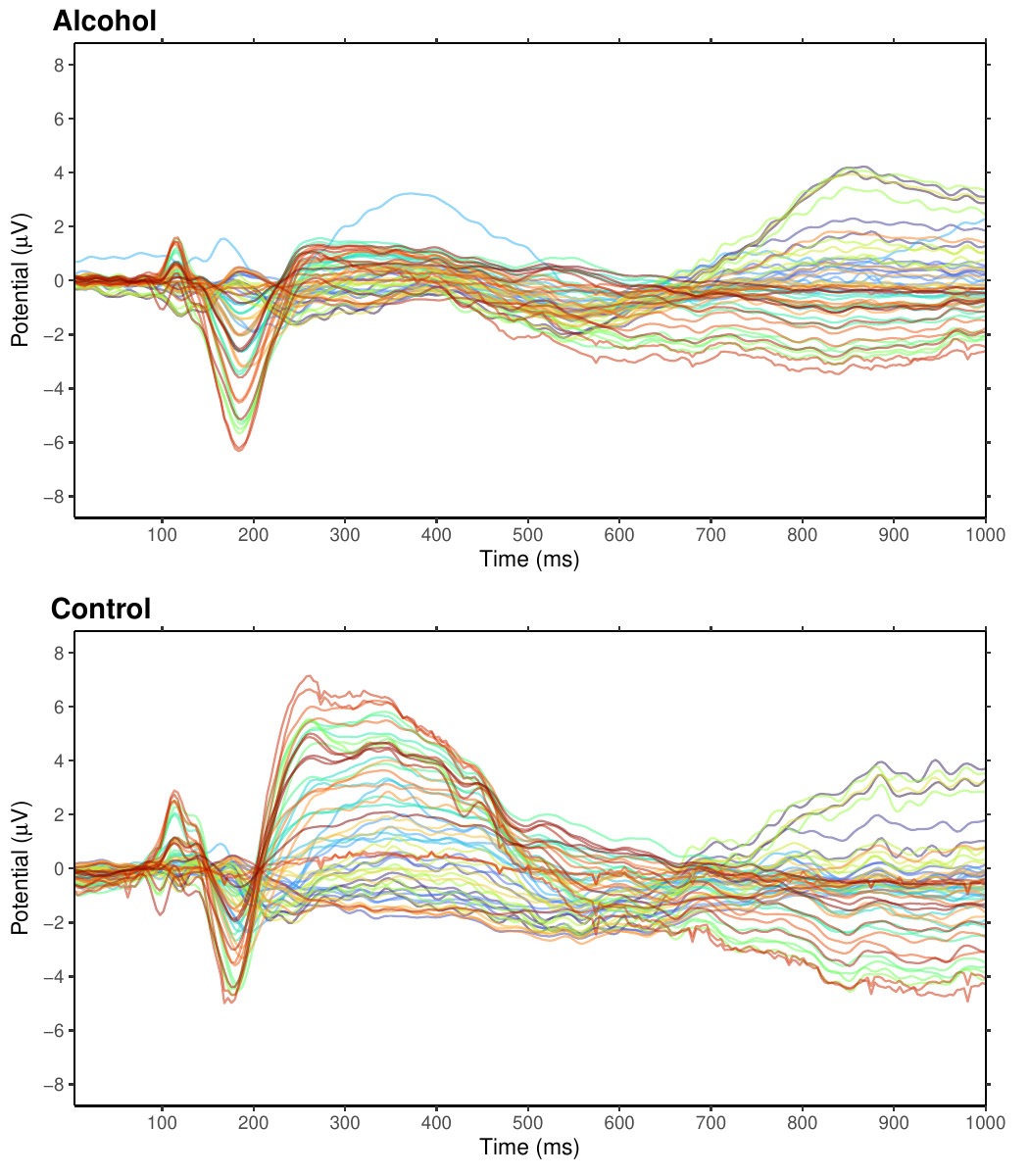}
    \caption{ERP plots for the alcohol and control group.}
    \label{fig:erp-plots-alc-ctrl}
\end{figure}

Figure \ref{fig:erp-plots-alc-ctrl} shows the ERP plots for the alcoholic group (77 subjects) and the control group (45 subjects). The two plots consist of 57 time series curves representing the average potential of subjects within a group at 57 locations over 256 time points. Each of the 57 locations correspond with the same unique colored curve in both plots. The scalp maps in Figure \ref{fig:scalpmaps-alc-ctrl} represent the distribution of the average EEG signals at three specified time points. Since we employ local models at each time point, visualizing the localized distribution of mean potential across all subjects at a specific time point is beneficial. Looking at Figure \ref{fig:erp-plots-alc-ctrl} and Figure \ref{fig:scalpmaps-alc-ctrl}, the magnitude of the average signals for the control group within the 225-500ms time interval are distinctively larger than that of the alcohol group, especially in the outer regions of the brain in locations illustrated with red shading on the scalp maps. The observed disparities serve as both inspiration and a crucial prerequisite, illuminating the potential not only to decipher functional connections between brain regions but also to classify subjects into either the alcohol or control group. Furthermore, the ERP plots and scalp maps in Figure \ref{fig:erp-plots-alc-ctrl} and Figure \ref{fig:scalpmaps-alc-ctrl} show that the set of locations with higher magnitudes of average potential varies over time. That is, different locations are active at different time points. 

\begin{figure}[!t]
    \centering
    \includegraphics[scale=0.55, trim=4.5cm 3.5cm 6.5cm 6.cm, clip]{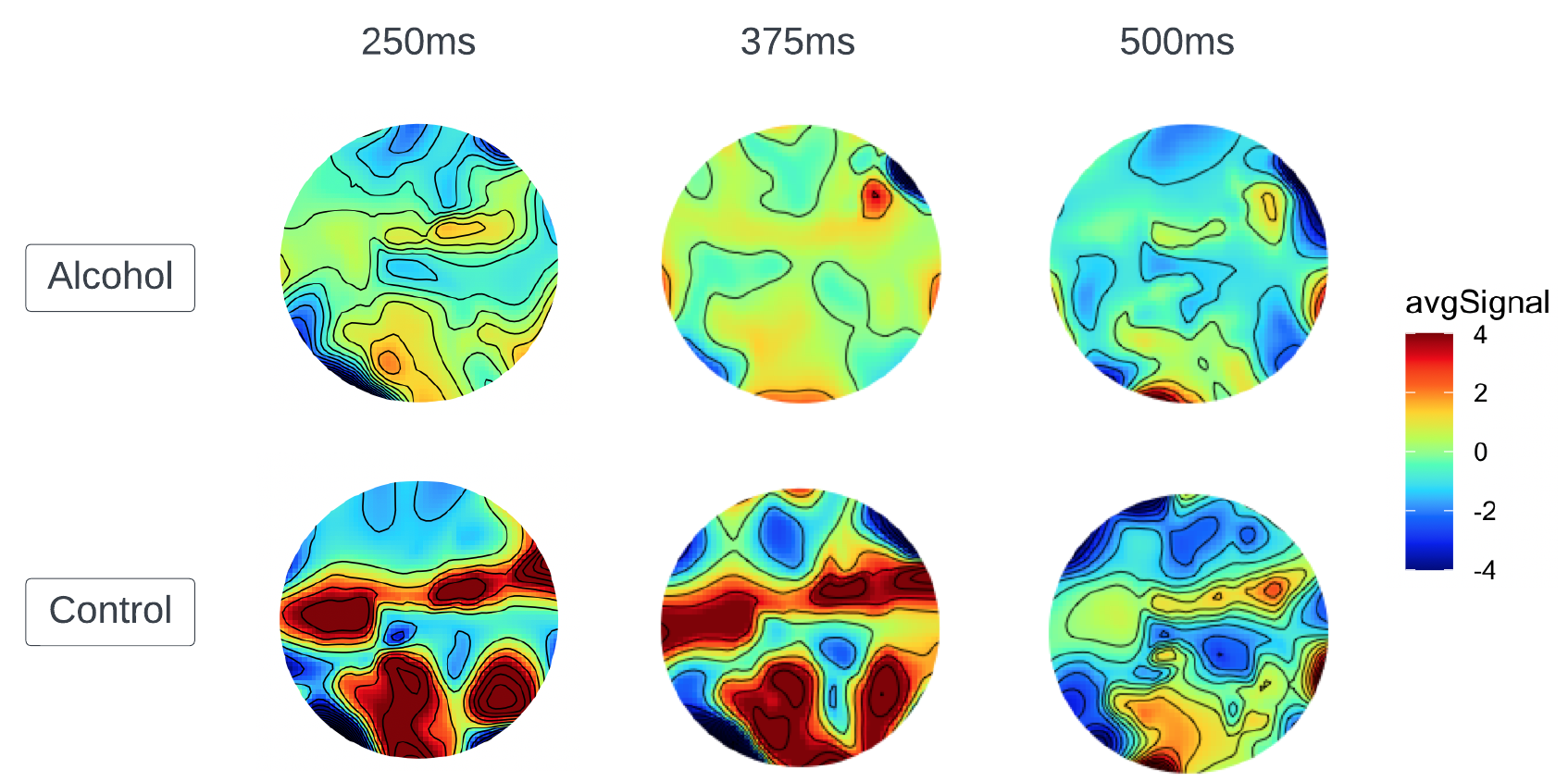}
    \caption{Scalp maps at 250ms, 375ms, and 500ms for the alcohol group and the control group.}
    \label{fig:scalpmaps-alc-ctrl}
\end{figure}

\section{Local Bayesian Modeling Framework}
\label{sec:lbm} 

As a way of bypassing the curse of dimensionality in high dimensional spatio-temporal data, we construct subject-level Bayesian classification models through logistic regression, locally at each time point. The GD prior is assigned to model coefficients, which correspond to locations. Although the GD prior inherently possesses shrinkage and selection properties, the GD prior cannot produce sparse solutions on its own. As a result, we implement a two-stage feature extraction process, with the first stage being a hard thresholding technique based on controlling the false discovery rate (FDR). The second stage consists of clustering locations based on the area under the curve computed from plotting the estimated coefficient vectors of each location. We aggregate the predictions of the local models over time to make final predictions. In the scope of this modeling perspective, we develop a general likelihood function, define the GD prior distribution, specify priors on the hyperparameters, and demonstrate the derivation of the posterior distribution.

\subsection{GD Prior Specification and Posterior Computation}
\label{sec:gdpriorspec}

Following the formulation of $\hat{\boldsymbol{\beta}}_{PL}$ in Eq. (\ref{plest}), we choose the wide-ranging class of Bregman divergence measures as our loss function. Depending on the form of the convex function, Bregman divergence measures reduce to well-known loss functions \citep{goh2018bayesian}. See Table A.1 in Section A.2.1 of the Annexure for examples. To develop a general likelihood function, we highlight the \textit{duality property} (explained in Section A.2.1 of the Annexure) between the loss function and the likelihood function. We define our likelihood function as $f_{\psi}(\boldsymbol{y} | \boldsymbol{\beta}) \propto \exp\left\{- BD_{\psi}(\boldsymbol{y}, \boldsymbol{h}(\boldsymbol{X\beta}))\right\}$,
with the assumption that $f_{\psi}(\boldsymbol{y} | \boldsymbol{\beta})$ is integrable. \cite{banerjee2005clustering} showed that there is a bijection between the natural exponential family and the class of Bregman divergence measures. The flexibility of the class of Bregman divergence measures allows us to handle a variety of response variables and likelihood functions. See Table A.2 in Section A.2.1 of the Annexure for examples. With the EEG data analysis, one of our goals is binary classification of subjects. Thus, we use the Bregman divergence measure uniquely defined for the Bernoulli family of distributions, $z_1\log(\frac{z_1}{z_2}) + (1 - z_1)\log(\frac{1 - z_1}{1 - z_2})$, which leads to our likelihood function being proportional to the Bernoulli likelihood function (see Section A.2.1 of the Annexure).

The GD prior is defined as follows \citep{goh2018bayesian} 
\begin{gather}
\pi(\boldsymbol{\beta}, \boldsymbol{d}) \propto \pi_G(\boldsymbol{\beta} | \boldsymbol{d}) \pi_D(\boldsymbol{d}), \label{gd} \\
\pi_G(\boldsymbol{\beta} | \boldsymbol{d}) \propto \prod_{j=1}^L \left[d_j^{1/2} \exp\left\{-\frac{d_j}{2} \beta_j^2\right\}\right], \label{g} \\
\pi_D(\boldsymbol{d}) \propto \prod_{j=1}^L \left[d_j^{\lambda - 1/2} \exp\left\{-\frac{\tau_0^2}{2\{(\boldsymbol{X}_{[, j]})^T\boldsymbol{X}_{[, j]}\}} d_j\right\}\right], \label{d}
\end{gather}
where $\tau_0 > 0$ is chosen to be sufficiently small. A critical piece of the GD prior is that it leads to a continuous and differentiable approximation of the $\ell_0$-norm penalty, which acts as the penalty function in Eq. (\ref{plest}). \cite{goh2018bayesian} introduced an $\ell_0$-norm approximation, referred to as the GD penalty, $\tilde{\ell}_0(\boldsymbol{\beta} | \lambda, \tau_0) = \lambda \sum_{j = 1}^L \frac{\{(\boldsymbol{X}_{[, j]})^T\boldsymbol{X}_{[, j]}\}\beta_j^2}{\tau_0^2 + \{(\boldsymbol{X}_{[, j]})^T\boldsymbol{X}_{[, j]}\}\beta_j^2}$, where $\boldsymbol{X}_{[, j]}$ denotes the $j^{th}$ column of the design matrix $\boldsymbol{X}$. We can see that for given values of the tuning parameter $\lambda$ and coefficient vector $\boldsymbol{\beta}$, $\tilde{\ell}_0(\boldsymbol{\beta} | \lambda, \tau_0) \rightarrow \lambda \sum_{j=1}^L \mathbbm{1}_{(\beta_j \neq 0)}$ as $\tau_0 \rightarrow 0$. Furthermore, $\tilde{\ell}_0(\boldsymbol{\beta} | \lambda, \tau_0$) is continous and differentiable everywhere. Of the many shrinkage and selection operators developed to induce sparsity into the PL estimator, \cite{dicker2013variable} introduced a continuous approximation of the $\ell_0$-norm penalty called the seamless-$\ell_0$ (SELO) penalty. However, it has restrictions in that it is not differentiable everywhere, lending an advantage to the GD penalty over the SELO penalty. We define the SELO penalty and provide a numerical demonstration of the advantage the GD penalty has over the SELO penalty in Section A.2.2 of the Annexure. A disadvantage of the GD penalty function is that it can't induce sparsity into the estimator unless $\tau_0 = 0$. We overcome this difficulty through our two-stage feature extraction process. 

In our GD approach, the hyperparameter $\lambda$ governs the degree of sparsity in the estimator of $\boldsymbol{\beta}$. Prior work, such as \cite{goh2018bayesian}, selects $\lambda$ via a grid search based on marginal likelihood, which introduces substantial computational overhead and relies on a frequentist tuning procedure within an otherwise Bayesian model. We propose a fully Bayesian alternative by assigning a gamma prior to $\lambda$, with density $\pi(\lambda) \propto \lambda^{\alpha_1 - 1} \exp\{-\alpha_2 \lambda\}$,
where $\alpha_1 > 0$ and $\alpha_2 > 0$ are shape and rate parameters, respectively. This prior is chosen to be weakly informative---flat enough to avoid imposing strong assumptions, yet structured to reflect the practical insight that larger $\lambda$ induces greater sparsity. By centering the prior's mode near zero (but not at zero), we encourage moderate sparsity without collapsing all coefficients. This innovation eliminates the need for computationally expensive grid searches, preserves Bayesian coherence, and improves robustness by integrating uncertainty about $\lambda$ directly into posterior inference.

When latent variables are introduced into regression models, common practice is to integrate out the latent variable from the likelihood function and conduct inference only on the parameter of interest. In Bayesian regression models, the latent variable is integrated out of the posterior distribution and posterior inference proceeds only for the parameter of interest. Rather than integrating out $\boldsymbol{d}$, an essential component of our GD approach is to substitute $\boldsymbol{d}$ with its MAP estimator. Without this aspect of our approach, we lose the crucial connection between the GD prior and the GD penalty function described in Lemma 2, Remark 1, and their proofs in Section A.5 of the Annexure. Using these justifications, we can see that the GD estimator is a nice approximation of the PL estimator with a Bregman divergence measure as the loss function and $\tilde{\ell}_0(\cdot | \lambda, \tau_0)$ as the penalty function. \cite{tibshirani1996regression} pointed out that in a Bayesian framework, the PL estimator in Eq. (\ref{plest}) can be viewed as the MAP estimator of the posterior distribution given by $\pi(\boldsymbol{\beta} | \boldsymbol{y}) \propto f(\boldsymbol{y} | \boldsymbol{\beta}) \pi(\boldsymbol{\beta}) \propto \exp\left\{-[L(\boldsymbol{y}, h(\boldsymbol{X\beta})) + Pe(\boldsymbol{\beta} | \lambda)]\right\}$.

Since $\pi_{GD}(\boldsymbol{\beta}, \boldsymbol{d}) \propto \pi_G(\boldsymbol{\beta} | \boldsymbol{d})$ with respect to $\boldsymbol{\beta}$, the posterior distribution of $\boldsymbol{\beta}$ can be viewed as $\pi(\boldsymbol{\beta} | \boldsymbol{y}, \hat{\boldsymbol{d}}) \propto f_{\psi}(\boldsymbol{y} | \boldsymbol{\beta}) \pi_{GD}(\boldsymbol{\beta}, \hat{\boldsymbol{d}}) \propto f_{\psi}(\boldsymbol{y} | \boldsymbol{\beta}) \pi_G(\boldsymbol{\beta} | \hat{\boldsymbol{d}})$, where $\hat{\boldsymbol{d}}$ is the MAP estimator of $\boldsymbol{d}$. We outline the propriety of the GD posterior in Section A.2.3 of the Annexure. Due to the high dimensionality of our parameter space, we utilize the univariate full conditional distributions, given by $\pi(d_j | \text{others}) \propto d_j^{\lambda} \exp\{-\frac{\tau_0^2 / \{(\boldsymbol{X}_{[, j]})^T\boldsymbol{X}_{[, j]}\} + \beta_j^2}{2} d_j\}$ and $\pi(\beta_j | \text{others}) \propto \exp\{-BD_{\psi} (\boldsymbol{y}, \boldsymbol{h(X\beta)}) - \frac{d_j}{2}\beta_j^2\}$, $j = 1, \dots, L$. Note, it is possible to express $\pi(d_j | \text{others})$ in closed form, and due to the smoothness of $\pi(\beta_j | \text{others})$, we can implement the Newton-Rhapson algorithm to find its mode. Prior knowledge about $\beta_j$ is represented through $\beta_j | \hat{d}_j \stackrel{ind.}{\sim} N(0, 1/\hat{d}_j)$. Lemma 2 in Section A.5 of the Annexure shows the form of $\hat{d}_j$. The posterior distribution of lambda is $\pi(\lambda | others) \propto \lambda^{\alpha_1 - 1} \exp\{-(\alpha_2 - \sum_{j=1}^L log (d_j)) \lambda \}$. For the remainder of this section, we demonstrate how we utilize the Metropolis-Hastings algorithm within a component-wise Gibbs sampler to carry out the MCMC algorithm.  

\underline{MCMC sampling procedure}. We implement the MCMC algorithm to obtain samples from our target posterior. \cite{goh2018bayesian} implement a component-wise ICM algorithm within the MCMC algorithm to compute the MAP estimate of $\boldsymbol{\beta}$. For each iteration of the algorithm, the cutoff is defined as $\xi_j = \frac{2 \lambda}{\tau_0^2 / [(\boldsymbol{X}_[,j])^T \boldsymbol{X}_{[,j]}] + \hat{\beta}_j^2}$, $j = 1, \dots, L$, which corresponds to a $95\%$ credible interval when $\hat{\beta}_j = 0$. As a result of the hard thresholding, $\hat{\boldsymbol{\beta}}$ contains exact zeros. However, we found their approach induces strong sparsity into the estimator. Consequently, we propose an entirely Bayesian MCMC sampling approach, inducing sparsity into the estimator in the first stage of feature extraction.

\underline{Modified version of Metropolis-Hastings algorithm}. Due to the high dimensionality of the parameter space, we additionally implement a component-wise Gibbs sampler by iteratively generating $\beta_j$'s from their component-wise full conditionals. Since a closed form of the full conditional distribution of $\beta_j$ may not exist, we propose to use the Metropolis-Hastings algorithm within the Gibbs sampler, where the move from the current state $\beta_j^{(t)}$ to a new state $\beta_j^{(t+1)}$ is generated from the proposal density $p(\beta_j^{(t+1)} | \beta_j^{(t)})$. Our Metropolis-Hastings algorithm differs from \cite{goh2018bayesian}, because we don't implement the ICM algorithm. Let $\boldsymbol{\beta}^{(t)}_{(j)} = (\beta_1^{(t+1)}, \dots, \beta_{j-1}^{(t+1)}, \beta_j^{(t)}, \beta_{j+1}^{(t)}, \dots, \beta_L^{(t)})$ denote the coefficient vector at the current state and $\boldsymbol{\beta}^{(t+1)}_{(j)} = (\beta_1^{(t+1)}, \dots, \beta_{j-1}^{(t+1)}, \beta_j^{(t+1)}, \beta_{j+1}^{(t)}, \dots, \beta_L^{(t)})$ denote the coefficient vector at the next state. We define the proposal density to be Gaussian with mean $\beta_j^{(t)}$ and variance $\delta [\frac{\partial^2}{\partial \beta_j^2} BD_{\psi}(\boldsymbol{y}, \boldsymbol{h(X\beta)} ) |_{\boldsymbol{\beta} = \boldsymbol{\beta}^{(t)}_{(j, L)}} + \hat{d}_j ]^{-1}$, owing to the form of the full conditional of $\beta_j$. Here, $\delta > 0$ is a predetermined constant to control the acceptance rate. For all $1 \leq j \leq L$, the proposed moves from state to state are accepted with probability 
\begin{equation*}
    \alpha = \text{min} \left\{1, \frac{f_{\psi}\left(\boldsymbol{y} | \boldsymbol{\beta}^{(t+1)}_{(1, j)}\right)\pi_{GD}\left(\boldsymbol{\beta}^{(t+1)}_{(1, j)}\right) p\left(\beta_j^{(t)} | \beta_j^{(t+1)}\right)}{f_{\psi}\left(\boldsymbol{y} | \boldsymbol{\beta}^{(t)}_{(j, L)}\right) \pi_{GD}\left(\boldsymbol{\beta}^{(t)}_{(j, L)}\right) p\left(\beta_j^{(t+1)} | \beta_j^{(t)}\right)} \right\}.
\end{equation*}

\subsection{Two-Stage Feature Extraction}
\label{sec:twostage}

A major goal of our Bayesian feature extraction process is to identify active regions of the brain according to the magnitude of their electrical signal. We propose a two-stage feature extraction process to obtain final estimates and select active locations. Although the MCMC sampling approach doesn't produce exact 0 estimates, the inherent feature extraction nature of the GD prior shrinks coefficients of smaller magnitude, while preserving coefficients of larger magnitude. Sparsity is induced into the estimator through a false discovery rate (FDR) thresholding approach in stage one. Coefficients which were preserved through shrinkage will remain preserved, whereas those which were reduced toward 0 will be dropped. Since the covariates in the local models are locations, performing the first stage in this way accounts for the fact that the set of locations selected is not uniform over all time points. In the second stage, we use a k-means clustering approach on the area under the curve formed by the estimated coefficient vectors of locations obtained from stage one. 

\underline{First-stage feature extraction with FDR control}.
The first stage of the feature extraction process is a FDR-based thresholding approach involving Bayesian model averaging and MAP estimates \citep{morris2008bayesian,mohammed2021radio}. We consider $N$ MCMC samples of $\boldsymbol{\beta}_t = (\beta_{1t}, \dots, \beta_{Lt})^T$, $\{\boldsymbol{\beta}_t^{(s)} : s = 1, \dots, N\}$, for $t = 1, \dots, \tau$ from the GD posterior distribution. For each location $l = 1, \dots, L$, we can calculate posterior inclusion probabilities $p_{lt} = \frac{1}{N} \sum_{s=1}^N \mathbbm{1}_{(|\beta_{lt}^{(s)}| \leq c)}$ from the observed MCMC samples for some constant $c > 0$. This pooled statistic can be interpreted as the local FDR \citep{morris2008bayesian}, and $(1 - p_{lt})$ is the probability that location $l$ at time point $t$ is active. For coefficients $\beta_{lt}$ close to 0, $p_{lt}$ is large so that the probability of the coefficient being significant is small. On the other hand, $p_{lt}$ is small for coefficients of higher magnitude, indicating a high probability of the coefficient being significant. The major question is, what values of $p_{lt}$ are considered small, and what values are considered large? We choose a threshold $\psi_{\alpha}$ such that the overall average FDR is controlled at level $\alpha$. The coefficient $\beta_{lt}$ is said to be significant if $p_{lt} < \psi_{\alpha}$. By controlling the overall average FDR at level $\alpha$, we only allow for $100\alpha\%$ of the pairs in $\{(l,t) : p_{lt} < \psi_{\alpha}\}$ to be false-positive inclusions. The threshold $\psi_{\alpha}$ is computed by first sorting the posterior inclusion probabilities $p_{lt}$ over all locations and time points. We denote the sorted probabilities by $p_{(k)}$, for $k = 1, \dots, L\tau$, and let $u = \{K | \frac{1}{K} \sum_{k = 1}^K p_{(k)} \leq \alpha \}$. We define $\psi_{\alpha} = p_{(u)}$, and according to our approach at level $\alpha$, the set $\{(l, t) : p_{lt} < \psi_{\alpha}\}$ is considered active. 

\underline{Second-stage feature extraction with k-means clustering}. After the first stage, we have a new vector of coefficient estimates at each time point, $\hat{\boldsymbol{\beta}}^*_{.t}$, which contains exact zero values. Looking at all time points, we have a matrix of zero and nonzero coefficient estimates, say $\hat{\boldsymbol{\beta}}^*$. Specifically, $\hat{\beta}_{lt}^* = 0$ implies location $l$ at time point $t$ is inactive, and $\hat{\beta}_{lt}^* \neq 0$ implies location $l$ at time point $t$ is active. In the second stage of the feature extraction process, we are focusing on the individual locations, represented by rows of $\hat{\boldsymbol{\beta}}^*$, say $\hat{\boldsymbol{\beta}}^*_{l.}$. As a consequence of inducing sparsity into the estimator in the first stage, active locations will have distinctively more non-zero entries than inactive locations. To extract active locations, we start by creating a line plot where the x-axis represents each $\hat{\beta}^*_{lt}$ over time and the y-axis represents the magnitude of the estimates. From here, we compute the area of the geometric shapes created between each pair of points in the line plot. We refer to the area under the curve formed by $\hat{\boldsymbol{\beta}}^*_{l.}$ as the sum of the areas of the geometric shapes. The area under the curve formed by active locations will be notably larger than that of inactive locations. Active locations will demonstrate longer streaks of successive nonzero estimates over time. Conversely, inactive locations will show a flatter pattern, emerging from several streaks of 0 over time. Locations with substantially large magnitude coefficients for only a short period of time will also have larger areas under the curve. Finally, we apply k-means clustering to the areas under curve formed by location vectors. Although k-means is sensitive to outliers, particularly when $k = 2$, there are no outliers after the first stage. We conclude locations in the cluster with the larger centroid are active. Since the first stage initiates separation between active and inactive locations, k-means with $k = 2$ clusters is more likely to converge to centroids that accurately represent the centers of the two clusters.

\subsection{Subject-Level Prediction}
\label{sec:pred}

Considering we construct models locally at time points, each local model is able to produce its own predicted binary response. Having all $\tau$ models predict the same binary response is extremely improbable. At each time point, the $i$th subject is assigned a prediction probability, $\hat{p}_{it} = h(\boldsymbol{x}_{it}^T \hat{\boldsymbol{\beta}}_t)$, where $\hat{\boldsymbol{\beta}}_t$ is the estimated coefficient vector at time point $t$, and represents the likeliness of the $i$th subject classified as a 1 over a 0. Specifically, $\hat{p}_{it}$ close to 1 is strong indication to classify the subject as $\hat{y}_{it} = 1$ at time point $t$ and $\hat{p}_{it}$ close to 0 is strong indication to classify the subject as $\hat{y}_{it} = 0$ at time point $t$. 

Final predictions for each subject in such scenarios have been made using the length of longest run of the local predictions $\hat{y}_{it}$ \citep{mohammed2019bayesian}. However, by assigning subject $i$ as 1 with probability $\hat{p}_{it}$ and 0 with probability $1 - \hat{p}_{it}$, additional randomness is incorporated into the prediction process, hampering evaluations based on longest runs. For example, if $\hat{p}_{it} = 0.8$, then there is still $0.2$ probability of subject $i$ being classified as a 0. The effectiveness of the length of the longest run is significantly compromised by misclassification in local predictions. More so, clear differences in electrical signal, between the control and alcohol group, exist within a smaller interval of time than the final stretch of time where the two groups show extremely similar behavior (see Figure \ref{fig:erp-plots-alc-ctrl}). 

Instead, we construct local weights, $w_{it}$, with the goal of allocating higher priority to local models with more certainty of classifying the subject as a 0 or 1. Local models corresponding to values of $\hat{p}_{it}$ closer to either extreme are granted increased emphasis over those corresponding to $\hat{p}_{it}$ in the vicinity of 0.5. Local models yielding $\hat{p}_{it}$ near 0.5 have a high degree of uncertainty resulting from the considerable amount of noise debilitating the electrical signals. These time points are practically negligible and should not play a major role in the prediction process. For all $t$, the weights for subject $i$ are defined as $w_{it} = \frac{(\hat{p}_{it} - 0.5)^2}{\sum_{t=1}^{\tau} (\hat{p}_{it} - 0.5)^2}$.
Owing to the heterogeneity in the temporal cross-section of the data and accounting for the subject-to-subject variability in the data, the final predicted response for subject $i$ is $\hat{y}_i = 1$ if $\sum_{t=1}^{\tau}w_{it}\hat{p}_{it} > 0.5$; and $\hat{y}_i = 0$ otherwise. In addition to decreasing the leverage of the local models at which the signal is insignificant, our approach accounts for the fact that there are time points where the selected locations show minimal electrical activity.

\section{Numerical Studies}
\label{sec:numstud}

In this section, we investigate the performance of our GD approach in estimation, feature extraction, and subject-level prediction. Results of our GD method are compared to existing shrinkage and selection operators employed in PL regression: LASSO, adaptive LASSO (a-LASSO), elastic net (E-Net), mini-max concave penalty (MCP), smoothly clipped absolute deviation (SCAD). We include details surrounding the run time of the GD model in Section \ref{sec:conclusion}.

For all analyses reported in this paper, the model was fitted using the \texttt{R} package \texttt{rstan} \citep{stan2023rstan}; the Gibbs sampler derivation is provided in Section \ref{sec:gdpriorspec} for methodological context. The \texttt{rstan} implementation employs Hamiltonian Monte Carlo (HMC) \citep{betancourt2015hamiltonian, neal2011mcmc}. The MH-within-Gibbs sampler presented in Section \ref{sec:gdpriorspec} and the Stan HMC implementation target the same GD posterior distribution; they differ only in the mechanism used to generate proposals. HMC uses gradient information from the log-posterior to construct efficient proposals that traverse the posterior geometry in fewer steps than a random-walk MH scheme. The smoothness and differentiability of the GD posterior, a property that follows directly from the continuous and differentiable nature of the GD penalty (see Section~A.2.2), makes it particularly amenable to gradient-based sampling. The \texttt{rstan} implementation was adopted for all reported results on account of its superior computational efficiency, particularly for larger simulation settings and the EEG case study.

\underline{Prior elicitation}. A generalized least squares estimator of $\boldsymbol{\beta}$ is used as the initial value in the simulations and case study. We use $\tau_0 = 10^{-5}$, because it is sufficiently small and works well in the numerical studies. In consideration of the low signal-to-noise ratio in the data, the constant used in the FDR approach to calculate the posterior inclusion probabilities is chosen to be $c = 10^{-3}$. A sensitivity analysis across a small range of $c$ values is reported in Supplementary Section~\ref{sec:simsettings}. The FDR approach will not recognize any of the coefficients as significant for too large of a value of $c$, which results in the first stage of feature extraction producing an overly sparse coefficient matrix. We choose the shape and rate hyperparameters of the gamma distribution on $\lambda$ to be $\alpha_1 = 0.1$ and $\alpha_2 = 0.2$, respectively. In other words, the prior mean is 0.5 and the prior variance is 2.5. We have prior belief that the model tends to perform better when $\lambda$ is between 0 and 1, which is why we chose the prior mean to be 0.5. However, choosing a fairly large prior variance allows for more flexibility in the way we present our prior knowledge and allows for a wider exploration of the parameter space. 

\subsection{Simulation Study}
\label{sec:simstudy}

In order to demonstrate the effectiveness and validity of our GD approach, we fit our model to multiple simulated data sets with different settings. The simulated data sets are sparse, high dimensional and aim to resemble the data structure used in the case study. In longitudinal simulation studies, only a single time point is considered, so one can simply generate the data, $\boldsymbol{X}_t$, which can be used along with the true coefficient vector, $\boldsymbol{\beta}_t$, and known inverse link function, $h(\cdot)$, to generate the response vector $\boldsymbol{y}$. A major challenge in simulating data for our framework is that $\boldsymbol{X}_t$ and $\boldsymbol{\beta}_t$ vary over time, while the response vector remains the same, i.e., $\boldsymbol{y}_t = \boldsymbol{y}$, for all $t = 1, \dots, \tau$. See Section A.3 of the Annexure for details on how we simulate our data and coefficients. 

Since we want to assess the selection of active locations, we report the proportion of correctly and incorrectly selected locations using true positive rate (TPR) and false positive rate (FPR). We report the root mean square error (rMSE) of estimation after the first stage to illustrate the models ability to accurately estimate coefficients, although, sparsity in the final estimator causes estimation error rates to inflate. With $p = L \times \tau$, we define $rMSE = \sqrt{|| \boldsymbol{\beta} - \hat{\boldsymbol{\beta}} ||^2/p}$. We present the receiver operating characteristics (ROC) curves and the area under the ROC curve (AUC) values as well. Additionally, we report three measures of prediction error rates: TPR, FPR, and prediction error (PE).

Three different values of $L$ (25, 60, 75) and three different values of $\tau$ (100, 175, 250) were considered, resulting in nine total simulation settings. In addition, we consider three different $\sigma_{\text{noise}}$ values, 1, 1.5, and 2, where $\sigma_{\text{noise}}$ is the standard deviation of the normal distribution used to generate random noise. Let $\text{SNR} = (\max|\beta_j|)/\sigma_{\text{noise}}$ denote the signal-to-noise ratio, and $\max|\beta_j|$ represents the largest absolute regression coefficient. A larger value of $\sigma_{\text{noise}}$ decreases the SNR, whereas a smaller value increases it. Consequently, simulation results depend on the chosen SNR, which is controlled through $\sigma_{\text{noise}}$. In the analysis reported below, we consider $\text{SNR} = 1/1.5$. See Section A.3 of the Annexure for more details on the different simulation settings. Each setting was replicated 100 times. All of the rates in the tables below are averages, along with standard deviations, over the 100 replications. Table \ref{tab:simres1} contains results from estimation and feature extraction, and Table \ref{tab:simres2} contains results from subject-level prediction. 

\begin{table}[!t]
\centering
\caption{Estimation and feature extraction - Simulation results.} 
\begin{tabular}{ccccc}
\hline 
$L$ & $\tau$ & rMSE & corrInd & incorrInd \\
\hline
25 & 100 & 0.542 (0.009) & 0.860 (0.165) & 0.110 (0.122) \\
& 175 & 0.541 (0.006) & 0.946 (0.110) & 0.057 (0.079) \\
& 250 & 0.541 (0.005) & 0.952 (0.090) & 0.038 (0.052) \\
\hline
60 & 100 & 0.574 (0.037) & 0.676 (0.165) & 0.254 (0.101) \\
& 175 & 0.575 (0.030) & 0.778 (0.158) & 0.271 (0.100) \\
& 250 & 0.571 (0.024) & 0.812 (0.125) & 0.234 (0.100) \\
\hline
75 & 100 & 0.966 (0.142) & 0.611 (0.138) & 0.334 (0.083) \\
& 175 & 0.971 (0.092) & 0.691 (0.139) & 0.322 (0.098) \\
& 250 & 0.979 (0.072) & 0.735 (0.131) & 0.342 (0.104) \\
\hline
\end{tabular}
\label{tab:simres1}
\end{table} 

As the number of locations, $L$, increases, the accuracy of correctly selected indices decreases, while the proportion of incorrectly selected locations rises. This trend makes sense since the number of predetermined active locations increases with $L$. Conversely, as the number of time points, $\tau$, increases, the proportion of correctly selected locations increases, and the proportion of incorrectly selected locations decreases. Since the length of the location vectors is $\tau$, for a large $\tau$ , the frequency of 1's will be much greater than the frequency of 0's in predetermined active locations, making it easier for the second stage of feature extraction to differentiate between active and inactive locations. Estimation accuracy also decreases with larger $L$, which is not surprising as the number of coefficients in each local model is $L$. Taking into account that the MCMC sampling algorithm is conducted at each local model, increasing $L$ makes it more difficult for the algorithm to converge. There is no clear pattern that arises in the rMSE over $\tau$, which is understandable considering the dimension of the local coefficient vectors is $L$ and irrelevant of $\tau$.  

The general trends observed in the results for prediction follow similar to those for estimation and feature extraction. As $L$ increases, prediction accuracy tends to decrease, and as $\tau$ increases, prediction accuracy increases. If the coefficient estimates are less accurate for large $L$, then there will also be more variance in the local prediction probabilities. By conducting the weighted prediction process as described in Section \ref{sec:pred}, the prediction accuracy decreases when there is more variability in local prediction probabilities. In the same way, we anticipate any additional uncertainty in the prediction process to decrease when $\tau$ is larger. Overall, the model appears to perform the best in estimation and feature extraction for large values of $\tau$ and small values of $L$. For prediction, the model appears to perform the best for large values of $\tau$, mostly irrelevant of $L$. In the real data, there are 57 locations and 256 time points. Given the findings in the simulation setting with $L = 60$ and $\tau = 250$, we expect the model to perform fairly well in the EEG case study. We understand that the high subject-to-subject variability and low signal-to-noise ratio in the real data will likely reduce overall performance, but we anticipate results comparable to previous methods. Tables A.3-A.6 in Section A.4 of the Annexure contain results for different values of SNR.

\begin{table}[!t]
\centering
\caption{Prediction rates and AUC - Simulation results.}
\begin{tabular}{cccccc}
\hline
$L$ & $\tau$ & TPR & FPR & PE & AUC \\
\hline
25 & 100 & 0.908 (0.076) & 0.105 (0.098) & 0.096 (0.063) & 0.963 (0.041) \\
& 175 & 0.953 (0.052) & 0.037 (0.065) & 0.043 (0.038) & 0.993 (0.011) \\
& 250 & 0.979 (0.040) & 0.025 (0.044) & 0.022 (0.029) & 0.998 (0.006) \\
\hline
60 & 100 & 0.879 (0.092) & 0.123 (0.116) & 0.123 (0.079) & 0.951 (0.059) \\
& 175 & 0.945 (0.055) & 0.062 (0.080) & 0.057 (0.049) & 0.984 (0.032) \\
& 250 & 0.968 (0.043) & 0.036 (0.055) & 0.033 (0.036) & 0.993 (0.018) \\
\hline
75 & 100 & 0.884 (0.083) & 0.113 (0.094) & 0.114 (0.066) & 0.955 (0.043) \\
& 175 & 0.935 (0.058) & 0.055 (0.071) & 0.061 (0.048) & 0.986 (0.023) \\
& 250 & 0.979 (0.036) & 0.026 (0.053) & 0.023 (0.032) & 0.996 (0.012) \\
\hline
\end{tabular}
\label{tab:simres2}
\end{table} 

\subsection{EEG Case Study}
\label{sec:eegstudy}

In this section, we provide results obtained from applying our GD method to a multi-subject EEG case study. Our goal is to classify subjects according to their EEG measurements to predict genetic predisposition to alcoholism. Details surrounding run time of the model are in Section \ref{sec:conclusion}. Due to the low signal-to-noise ratio and several other complex characteristics of the data, high prediction accuracy is hard to achieve. In light of this, we report ROC curves and AUC values on top of the prediction error rates which were used in the simulation studies. Results from the competing methods (LASSO, a-LASSO, E-Net, MCP, SCAD) are obtained by vectorizing the model coefficients into a $L\tau$-dimensional vector and matricizing the three-dimensional array of data into an $n \times L\tau$-dimensional matrix. On the other hand, our GD method constructs local models which actually tackles the high dimensional problem and allows us to recognize spatial and temporal patterns observed in the data. Table \ref{tab:eegres1} provides the results and Figure \ref{fig:roccurves-sameplot} contains ROC curves for each method. We do not include comparisons with the spike-and-slab approach in \cite{mohammed2019bayesian}, because a function from an R package that was used to fit the model is no longer supported. For context, \cite{mohammed2019bayesian} reported a prediction accuracy of 70.03\% (equivalently, a prediction error of 29.97\%) with TPR $= 0.7374$ and FPR $= 0.3600$ on the same EEG dataset using fivefold cross-validation with the Wcut thresholding method. Our GD method achieves an improved prediction error of 26.23\% (prediction accuracy of 73.77\%), with TPR $= 0.7143$ and FPR $= 0.2222$ using LOOCV. The two methods are broadly comparable in overall prediction accuracy, with our GD method achieving a lower prediction error and a notably lower FPR, suggesting that the GD prior with FDR-controlled screening induces more controlled sparsity in the selected locations.

\begin{table}[!t]
\centering
\caption{Prediction error rates and AUC - EEG results.}
\label{tab:eegres1}
\begin{tabular}{lcccc}
\hline
Method & TPR & FPR & PE & AUC \\
\hline
GD & 0.7143 & 0.2222 & 0.2623 & 0.7856 \\
LASSO & 0.8571 & 0.3778 & 0.2295 & 0.7616 \\
a-LASSO & 0.7662 & 0.2222 & 0.2295 & 0.7968 \\
E-Net & 0.8052 & 0.4000 & 0.2705 & 0.7723 \\
MCP & 0.8182 & 0.4444 & 0.2787 & 0.7752 \\
SCAD & 0.8571 & 0.3778 & 0.2295 & 0.7651 \\
\hline
\end{tabular}
\end{table}

\begin{figure}[!t]
    \centering
    \includegraphics[scale=0.6, trim=0.25cm 0.25cm 0.5cm 0.25cm, clip]{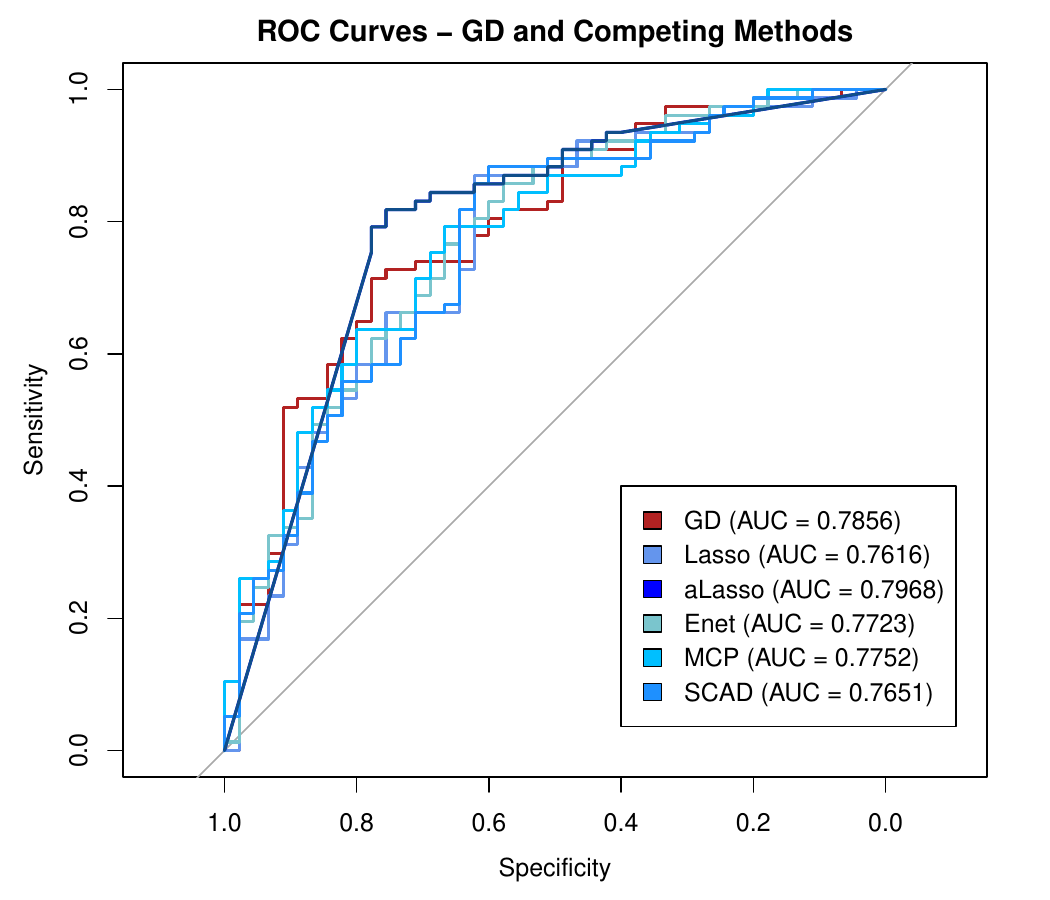}
    \caption{ROC curves - EEG results.}
    \label{fig:roccurves-sameplot}
\end{figure}

Our final subject-level predictions were based on the coefficient estimates obtained after the first-stage feature selection, $\hat{\boldsymbol{\beta}}^*$. Although all of the competing methods have high TPR, the GD method has the lowest FPR (tied with a-LASSO). Also, the GD method has the second highest AUC behind a-LASSO. Since the a-LASSO model was fit using ridge weights which were found using cross-validation, it makes sense that it performed the best here in terms of prediction rates. Even with that, the results from the GD method are comparable to a-LASSO and outperforms the other competing methods. Over and above that, our GD method determines which locations of the brain are active, identifying patterns of functional connectivity associated with cognitive or behavioral states. The competing methods can only determine whether an individual $(l, t)$ pair is significant, i.e., whether a coefficient is zero or nonzero. They are unable to determine active locations as our framework allows.

We implemented leave-one-out cross-validation (LOOCV). With $n = 122$ subjects in the data, LOOCV can be thought of as a $k$-fold cross-validation where $k = 122$. Using this approach, an average of approximately $3.87$ locations were selected over the 122 folds, which accounts for approximately $6.79\%$ of the 57 locations. Similar to what is seen in Figure 2 of \cite{hu2015local} and the plots in Section \ref{sec:eeg}, the bulk of the activity is detected at the outer locations of the brain. We found that majority of the folds resulted in the model selecting the same 3 electrodes, CP5, P7, and P07, as active. It is worth noting these electrodes are located in the parietal lobe of the brain, which plays a crucial role in a wide range of sensory and cognitive processes. Note that, to empirically validate the choice of number of clusters in the second-stage feature extraction with k-means clustering, we considered cluster sizes $k = 2$ and $k = 3$ on the full EEG data set and fit the local models. We found that the average silhouette widths were 0.21 and 0.02 for $k = 2$ and $k = 3$, respectively. Although 0.21 ($k=2$) reflects modest separation it is substantially higher than 0.02 ($k = 3$), supporting $k = 2$ as the empirically appropriate choice.

The parietal lobe is responsible for processes such as sensory integration and spatial awareness, and has been referred to as the association region of the brain \cite{ackerman1992discovering}. It has also been shown to be essential for visual tracking \citep{battelli2009role}. Aside from the parietal lobe, other electrodes selected multiple times were located in the frontal lobe, shown to have fast visual responses, thought to reflect mostly low-level visual processing, and delayed responses that correlate with perceptual reports \citep{libedinsky2011role}. Table \ref{tab:eeglocsel} provides the names of electrodes along with their associated frequency of selection. Only $18$ out of $57$ of the locations were selected at least once. To demonstrate where these electrodes are located relative to different regions of the brain, we included a scalp map in Figure \ref{fig:scalpmap-selLoc} based on the frequency of selection for all locations over the $122$ folds. We used a square root transformation on the data because of the large range of frequencies from 0 to 122. The darker red seen on the bottom left demonstrates the activity in the parietal lobe, and the darker red on the upper left and top of the plot demonstrate the activity observed in different cortical areas of the frontal lobe. See Figure A.2 in Section A.4 of the Annexure for the exact spatial locations of the electrodes. 
\begin{table}[!t]
    \centering
    \caption{Electrodes frequency of selection over the 122 folds in LOOCV.}
    \begin{tabular}{cccccccccc}
        \hline
Electrode & FPZ & AFz & F7 & F5 & F3 & F8 & FT7 & FC3 & FCZ \\
Frequency & 39 & 2 & 7 & 6 & 2 & 1 & 20 & 1 & 3 \\
\hline
Electrode & CP5 & CP1 & CP4 & P7 & P5 & P4 & PO7 & POz & O2 \\
Frequency & 121 & 1 & 1 & 122 & 15 & 3 & 122 & 4 & 2 \\
\hline
    \end{tabular}
    \label{tab:eeglocsel}
\end{table}

\begin{figure}[!t]
    \centering
    \includegraphics[scale=0.6, trim=0.5cm 3.15cm 0.25cm 3cm, clip]{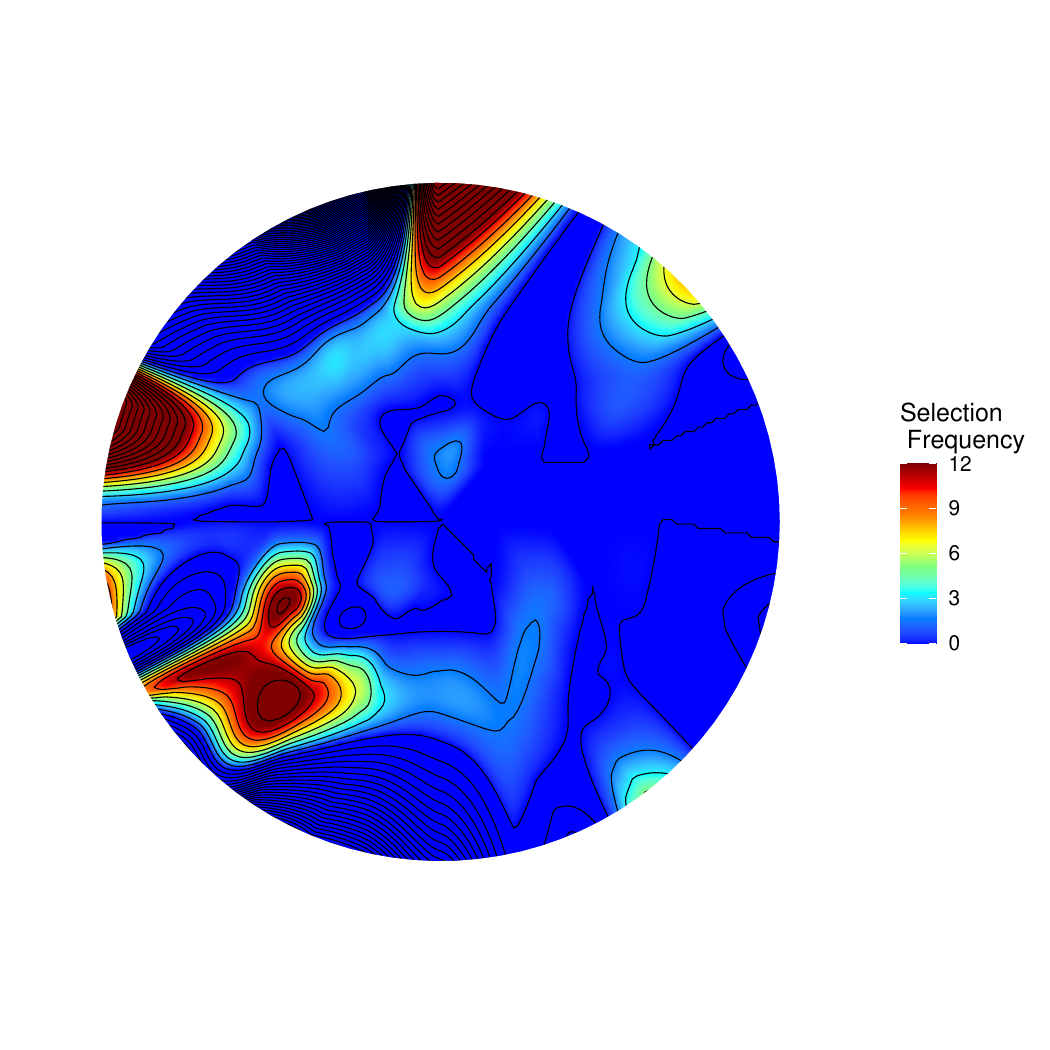}
    \caption{Scalp map showing the frequency (square root) of selection for the 57 locations.}
    \label{fig:scalpmap-selLoc}
\end{figure}

\section{Conclusion}
\label{sec:conclusion}

We introduce a novel Bayesian approach, utilizing a continuous-and-differentiable approximation of the $\ell_0$-norm penalty, to address the challenge of simultaneous feature extraction and estimation in handling sparse, high dimensional spatio-temporal data. To reduce the dimension of the problem, we employ a local modeling approach, offering the advantage of parallelizing local models. The Gaussian and gamma mixture representation allows for estimation via traditional Bayesian methods such as the MCMC sampling algorithm to extract information regarding the target posterior distribution. Although our method is being applied to EEG measurements to examine the effects of chronic alcohol exposure on specific brain regions, its adaptability extends to other domains with a similar spatio-temporal structure and dependencies across multiple dimensions.

Over and above the frequentist PL methods, Bayesian methods are able to quantify uncertainty in the PL estimator. We assign a gamma prior to the regularization parameter as a computationally efficient alternative to frequentist grid search, on top of other potential modeling benefits to explore. BD measures as loss functions allows our method to handle various likelihood functions. Our two-stage feature extraction process induces sparsity into the estimator and determines active locations. FDR thresholding induces sparsity in the first stage, while k-means clustering on the areas under the curve formed by estimated coefficient vectors selects active locations, avoiding a fixed threshold. While our GD approach doesn't explicitly incorporate spatial-temporal structure, it demonstrates computational efficiency. 

\underline{Spatio-temporal characteristics of the data and compute time}. Bypassing the spatio-temporal structure of the data in our local modeling approach is a strategic choice, allowing us to avoid the need of extensive computational power. We still acknowledge the spatio-temporal behavior of the data. In the second stage, we employ k-means clustering to group locations exhibiting similar behavior over time. We account for the varying signal-to-noise ratio at different time points, giving less weight to those with an extremely low signal-to-noise ratio during our weighted prediction process. Despite alternative approaches explicitly incorporating spatial-temporal structure, our method yields computational efficiency and comparable results. For example, fitting a single local model on the EEG data with $1000$ MCMC iterations, $n = 122$ subjects, and $L = 57$ locations on a computer with 32 GB of memory, 10 cores, and a Apple M1 Max chip, yields a run time of approximately $35.23$ seconds, on average, with a standard deviation and range of approximately $11.54$ seconds and $57.44$ seconds, respectively. Fitting all $\tau = 256$ local models individually yields a run time of approximately $150.33$ minutes. When we run 10 local models at a time, in parallel, on the 10 cores, the run time to fit all $\tau = 256$ local models is reduced to approximately $36.19$ minutes. To further characterize scaling with $n$ and $L$, we timed a single local model ($1{,}000$ MCMC iterations, $300$ warmup, one chain) for $(n, L) \in \{(100, 25),\,(100, 60),\,(200, 25)\}$ on an Intel Xeon E5-2660v3 node with 256 GB RAM, yielding mean sampling times of $13.2$, $26.5$, and $18.0$ seconds, respectively. Run time approximately doubles when $L$ increases from $25$ to $60$ with $n$ fixed, while increasing $n$ from $100$ to $200$ with $L$ fixed produces a more modest increase, indicating that run time is more sensitive to the number of locations than the number of subjects.

Although we avoided fixing a threshold, k-means clustering with two clusters comes with its own disadvantages. The weighted prediction process is based on a fixed threshold, leading to an arbitrary separation between subjects classified as alcoholic versus otherwise. Extending the framework to another real-data application (e.g., fMRI data or spatio-temporal climate measurements) which share the same three-dimensional array structure required by the method, remains an important and natural direction for future work. A major focus moving forward will be to improve the local modeling approach by incorporating the temporal and spatial structures explicitly through the prior. Beyond that, we also plan to extend the GD prior to three-dimensional array models, exploring new techniques to improve our feature extraction process, enhancing the prediction process, and assessing the model under flexible link functions.

{\parindent0pt
\newpage

\newpage
}


\section*{Acknowledgments}
SM was partially supported by Institutional Research Funds from Boston University and the Framingham Heart Study Brain Aging Program supported by U19-AG068753 from the National Institute of Health.
\vskip0.3cm

\section*{Code Availability Statement.} The code to implement the model using Stan is provided as an \texttt{R} package at
\url{https://github.com/garrett-frady/gdprior}.

\section*{Conflict of interest}
The authors do not have any financial or non-financial conflict of interest to declare for the research work included in this article.
\bibliographystyle{sa}
\bibliography{refs}

@article{goh2018bayesian,
  title     = {{B}ayesian {MAP} estimation using {G}aussian and diffused-gamma prior},
  author    = {Goh, Gyuhyeong and Dey, Dipak K},
  journal   = {Canadian Journal of Statistics},
  volume    = {46(3)},
  pages     = {399--415},
  year      = {2018},
  publisher = {Wiley Online Library}
}

@article{mohammed2019bayesian,
  title     = {{B}ayesian variable selection using spike-and-slab priors with application
               to high dimensional electroencephalography data by local modelling},
  author    = {Mohammed, Shariq and Dey, Dipak K and Zhang, Yuping},
  journal   = {Journal of the Royal Statistical Society: Series C (Applied Statistics)},
  volume    = {68(5)},
  pages     = {1305--1326},
  year      = {2019},
  publisher = {Wiley Online Library}
}

@article{banerjee2005clustering,
  title     = {Clustering with {B}regman divergences},
  author    = {Banerjee, Arindam and Merugu, Srujana and Dhillon, Inderjit S
               and Ghosh, Joydeep},
  journal   = {Journal of Machine Learning Research},
  volume    = {6(58)},
  pages     = {1705--1749},
  year      = {2005}
}

@article{tam2019human,
  author  = {Tam, Wing-kin and Wu, Tong and Zhao, Qi and Keefer, Edward
             and Yang, Zhi},
  year    = {2019},
  pages   = {12},
  title   = {Human motor decoding from neural signals: a review},
  volume  = {1},
  journal = {{BMC} Biomedical Engineering},
  doi     = {10.1186/s42490-019-0022-z}
}

@article{caffo2010twostage,
  title   = {Two-stage decompositions for the analysis of functional connectivity
             for {fMRI} with application to {A}lzheimer's disease risk},
  journal = {NeuroImage},
  volume    = {51(3)},
  pages   = {1140--1149},
  year    = {2010},
  doi     = {10.1016/j.neuroimage.2010.02.081},
  author  = {Caffo, Brian S and Crainiceanu, Ciprian M and Verduzco, Guillermo
             and Joel, Suresh and Mostofsky, Stewart H and Bassett, Susan Spear
             and Pekar, James J}
}

@article{zhou2013tensor,
  author  = {Zhou, Hua and Li, Lexin and Zhu, Hongtu},
  year    = {2013},
  pages   = {540--552},
  title   = {Tensor regression with applications in neuroimaging data analysis},
  volume  = {108},
  journal = {Journal of the American Statistical Association},
  doi     = {10.1080/01621459.2013.776499}
}

@article{zhou2014regularized,
  author  = {Zhou, Hua and Li, Lexin},
  year    = {2014},
  pages   = {463--483},
  title   = {Regularized matrix regression},
  volume  = {76},
  journal = {Journal of the Royal Statistical Society: Series B
             (Statistical Methodology)},
  doi     = {10.1111/rssb.12031}
}

@article{tian2012twoway,
  author  = {Tian, Tian Siva and Huang, Jianhua Z and Shen, Haipeng
             and Li, Zhimin},
  year    = {2012},
  title   = {A two-way regularization method for {MEG} source reconstruction},
  volume    = {6(3)},
  pages   = {1021--1046},
  journal = {Annals of Applied Statistics},
  doi     = {10.1214/11-AOAS531}
}

@article{hu2015local,
  author  = {Hu, Yue and Allen, Genevera I},
  year    = {2015},
  pages   = {905--917},
  title   = {Local-aggregate modeling for big-data via distributed optimization:
             applications to neuroimaging},
  volume    = {71(4)},
  journal = {Biometrics},
  doi     = {10.1111/biom.12355}
}

@inproceedings{lukic2002spatially,
  author    = {Lukic, A S and Wernick, M N and Hansen, L K and Anderson, J
               and Strother, S C},
  booktitle = {Proceedings {IEEE} International Symposium on Biomedical Imaging},
  title     = {A spatially robust {ICA} algorithm for multiple {fMRI} data sets},
  year      = {2002},
  pages     = {839--842},
  doi       = {10.1109/ISBI.2002.1029390}
}

@article{calhoun2001method,
  author  = {Calhoun, Vince D and Adali, T{\"u}lay and Pearlson, Godfrey D
             and Pekar, J J},
  year    = {2001},
  pages   = {140--151},
  title   = {A method for making group inferences from functional {MRI} data
             using independent component analysis},
  volume    = {14(3)},
  journal = {Human Brain Mapping},
  doi     = {10.1002/hbm.1048}
}

@article{svensen2002ica,
  author  = {Svens{\'e}n, Markus and Kruggel, Frithjof and Benali, Habib},
  year    = {2002},
  pages   = {551--563},
  title   = {{ICA} of {fMRI} group study data},
  volume  = {16},
  journal = {NeuroImage},
  doi     = {10.1006/nimg.2002.1122}
}

@article{beckman2005tensorial,
  title   = {Tensorial extensions of independent component analysis for
             multisubject {FMRI} analysis},
  journal = {NeuroImage},
  volume    = {25(1)},
  pages   = {294--311},
  year    = {2005},
  doi     = {10.1016/j.neuroimage.2004.10.043},
  author  = {Beckmann, C F and Smith, S M}
}

@article{wanmei2009distributed,
  title   = {A distributed spatio-temporal {EEG}/{MEG} inverse solver},
  journal = {NeuroImage},
  volume    = {44(3)},
  pages   = {932--946},
  year    = {2009},
  doi     = {10.1016/j.neuroimage.2008.05.063},
  author  = {Ou, Wanmei and H{\"a}m{\"a}l{\"a}inen, Matti S and Golland, Polina}
}

@article{huang2009analysis,
  author  = {Huang, Jianhua and Shen, Haipeng and Buja, Andreas},
  year    = {2009},
  pages   = {1609--1620},
  title   = {The analysis of two-way functional data using two-way regularized
             singular value decompositions},
  volume  = {104},
  journal = {Journal of the American Statistical Association},
  doi     = {10.1198/jasa.2009.tm08024}
}

@article{lee2010biclustering,
  title     = {Biclustering via sparse singular value decomposition},
  author    = {Lee, Mihee and Shen, Haipeng and Huang, Jianhua Z and Marron, J S},
  journal   = {Biometrics},
  volume    = {66(4)},
  pages     = {1087--1095},
  year      = {2010},
  publisher = {Wiley Online Library}
}

@article{witten2009penalized,
  title     = {A penalized matrix decomposition, with applications to sparse
               principal components and canonical correlation analysis},
  author    = {Witten, Daniela M and Tibshirani, Robert and Hastie, Trevor},
  journal   = {Biostatistics},
  volume    = {10(3)},
  pages     = {515--534},
  year      = {2009},
  publisher = {Oxford University Press}
}

@article{tibshirani1996regression,
  title     = {Regression shrinkage and selection via the {L}asso},
  author    = {Tibshirani, Robert},
  journal   = {Journal of the Royal Statistical Society: Series B
               (Methodological)},
  volume    = {58(1)},
  pages     = {267--288},
  year      = {1996},
  publisher = {Wiley Online Library}
}

@article{fan2001variable,
  title     = {Variable selection via nonconcave penalized likelihood and
               its oracle properties},
  author    = {Fan, Jianqing and Li, Runze},
  journal   = {Journal of the American Statistical Association},
  volume    = {96(456)},
  pages     = {1348--1360},
  year      = {2001},
  publisher = {Taylor \& Francis}
}

@article{zou2006adaptive,
  title     = {The adaptive {L}asso and its oracle properties},
  author    = {Zou, Hui},
  journal   = {Journal of the American Statistical Association},
  volume    = {101(476)},
  pages     = {1418--1429},
  year      = {2006},
  publisher = {Taylor \& Francis}
}

@article{akaike1974new,
  title     = {A new look at the statistical model identification},
  author    = {Akaike, Hirotugu},
  journal   = {{IEEE} Transactions on Automatic Control},
  volume    = {19(6)},
  pages     = {716--723},
  year      = {1974},
  publisher = {IEEE}
}

@article{schwarz1978estimating,
  title     = {Estimating the dimension of a model},
  author    = {Schwarz, Gideon},
  journal   = {The Annals of Statistics},
  pages     = {461--464},
  volume    = {6(2)},
  year      = {1978},
  publisher = {JSTOR}
}

@article{zou2005regularization,
  author  = {Zou, Hui and Hastie, Trevor},
  title   = {Regularization and variable selection via the elastic net},
  journal = {Journal of the Royal Statistical Society: Series B
             (Statistical Methodology)},
  volume    = {67(2)},
  pages   = {301--320},
  doi     = {10.1111/j.1467-9868.2005.00503.x},
  year    = {2005}
}

@article{kim2008smoothly,
  title     = {Smoothly clipped absolute deviation on high dimensions},
  author    = {Kim, Yongdai and Choi, Hosik and Oh, Hee-Seok},
  journal   = {Journal of the American Statistical Association},
  volume    = {103(484)},
  pages     = {1665--1673},
  year      = {2008},
  publisher = {Taylor \& Francis}
}

@article{dicker2013variable,
  title     = {Variable selection and estimation with the seamless-{$L_0$}
               penalty},
  author    = {Dicker, Lee and Huang, Baosheng and Lin, Xihong},
  journal   = {Statistica Sinica},
  pages     = {929--962},
  volume    = {23},
  year      = {2013},
  publisher = {JSTOR}
}

@article{zhang2010nearly,
  title     = {Nearly unbiased variable selection under minimax concave penalty},
  author    = {Zhang, Cun-Hui},
  journal   = {The Annals of Statistics},
  volume    = {38(2)},
  pages     = {894--942},
  year      = {2010},
  publisher = {Institute of Mathematical Statistics}
}

@article{casella2010penalized,
  title     = {Penalized regression, standard errors, and {B}ayesian {L}assos},
  author    = {Casella, George and Ghosh, Malay and Gill, Jeff and Kyung, Minjung},
  journal   = {Bayesian Analysis},
  volume    = {5(2)},
  pages     = {369--411},
  year      = {2010},
  publisher = {International Society for Bayesian Analysis}
}

@article{mitchell1988bayesian,
  title     = {{B}ayesian variable selection in linear regression},
  author    = {Mitchell, Toby J and Beauchamp, John J},
  journal   = {Journal of the American Statistical Association},
  volume    = {83(404)},
  pages     = {1023--1032},
  year      = {1988},
  publisher = {Taylor \& Francis}
}

@article{rovckova2014emvs,
  title     = {{EMVS}: the {EM} approach to {B}ayesian variable selection},
  author    = {Ro{\v{c}}kov{\'a}, Veronika and George, Edward I},
  journal   = {Journal of the American Statistical Association},
  volume    = {109(506)},
  pages     = {828--846},
  year      = {2014},
  publisher = {Taylor \& Francis}
}

@book{bernardo2009bayesian,
  title     = {{B}ayesian Theory},
  author    = {Bernardo, Jos{\'e} M and Smith, Adrian F M},
  series    = {Wiley Series in Probability and Statistics},
  year      = {2009},
  publisher = {Wiley}
}

@article{bregman1967relaxation,
  title     = {The relaxation method of finding the common point of convex sets
               and its application to the solution of problems in convex
               programming},
  author    = {Bregman, Lev M},
  journal   = {{USSR} Computational Mathematics and Mathematical Physics},
  volume    = {7(3)},
  pages     = {200--217},
  year      = {1967},
  publisher = {Elsevier}
}

@article{wedderburn1974quasi,
  title     = {Quasi-likelihood functions, generalized linear models, and the
               {G}auss--{N}ewton method},
  author    = {Wedderburn, Robert W M},
  journal   = {Biometrika},
  volume    = {61(3)},
  pages     = {439--447},
  year      = {1974},
  publisher = {Oxford University Press}
}

@article{zhang2009new,
  title     = {New aspects of {B}regman divergence in regression and
               classification with parametric and nonparametric estimation},
  author    = {Zhang, Chunming and Jiang, Yuan and Shang, Zuofeng},
  journal   = {Canadian Journal of Statistics},
  volume    = {37(1)},
  pages     = {119--139},
  year      = {2009},
  publisher = {Wiley Online Library}
}

@article{delorme2004eeglab,
  title     = {{EEGLAB}: an open source toolbox for analysis of single-trial
               {EEG} dynamics including independent component analysis},
  author    = {Delorme, Arnaud and Makeig, Scott},
  journal   = {Journal of Neuroscience Methods},
  volume    = {134(1)},
  pages     = {9--21},
  year      = {2004},
  publisher = {Elsevier}
}

@article{snodgrass1980standardized,
  title     = {A standardized set of 260 pictures: norms for name agreement,
               image agreement, familiarity, and visual complexity},
  author    = {Snodgrass, Joan G and Vanderwart, Mary},
  journal   = {Journal of Experimental Psychology: Human Learning and Memory},
  volume    = {6(2)},
  pages     = {174--215},
  year      = {1980},
  publisher = {American Psychological Association}
}

@article{besag1986statistical,
  title     = {On the statistical analysis of dirty pictures},
  author    = {Besag, Julian},
  journal   = {Journal of the Royal Statistical Society: Series B
               (Methodological)},
  volume    = {48(3)},
  pages     = {259--279},
  year      = {1986},
  publisher = {Wiley Online Library}
}

@article{mohammed2021radio,
  author    = {Mohammed, Shariq and Bharath, Karthik and Kurtek, Sebastian
               and Rao, Arvind and Baladandayuthapani, Veerabhadran},
  title     = {{RADIOHEAD}: radiogenomic analysis incorporating tumor
               heterogeneity in imaging through densities},
  volume    = {15(4)},
  journal   = {The Annals of Applied Statistics},
  publisher = {Institute of Mathematical Statistics},
  pages     = {1808--1830},
  year      = {2021},
  doi       = {10.1214/21-AOAS1458}
}

@article{hoeting1999bayesian,
  author    = {Hoeting, Jennifer A and Madigan, David and Raftery, Adrian E
               and Volinsky, Chris T},
  title     = {{B}ayesian model averaging: a tutorial},
  volume    = {14(4)},
  journal   = {Statistical Science},
  publisher = {Institute of Mathematical Statistics},
  pages     = {382--417},
  year      = {1999},
  doi       = {10.1214/ss/1009212519}
}

@article{morris2008bayesian,
  author  = {Morris, Jeffrey S and Brown, Philip J and Herrick, Richard C
             and Baggerly, Keith A and Coombes, Kevin R},
  year    = {2008},
  pages   = {479--489},
  title   = {{B}ayesian analysis of mass spectrometry proteomic data using
             wavelet-based functional mixed models},
  volume  = {64},
  journal = {Biometrics},
  doi     = {10.1111/j.1541-0420.2007.00895.x}
}

@misc{stan2023rstan,
  title  = {{RStan}: the {R} interface to {Stan}},
  author = {{Stan Development Team}},
  note   = {R package version 2.26.23},
  year   = {2023},
  url    = {https://mc-stan.org/}
}

@book{ackerman1992discovering,
  title     = {Discovering the Brain},
  author    = {Ackerman, Sandra},
  year      = {1992},
  publisher = {National Academies Press}
}

@incollection{betancourt2015hamiltonian,
  title     = {{H}amiltonian {M}onte {C}arlo for hierarchical models},
  author    = {Betancourt, Michael J and Girolami, Mark},
  booktitle = {Current Trends in {B}ayesian Methodology with Applications},
  editor    = {Upadhyay, Satyanshu K and Singh, Umesh and Dey, Dipak K
               and Loganathan, Appaia},
  chapter   = {4},
  pages     = {79--102},
  publisher = {Chapman and Hall/{CRC} Press},
  address   = {Boca Raton, FL},
  year      = {2015}
}

@incollection{neal2011mcmc,
  title     = {{MCMC} using {H}amiltonian dynamics},
  author    = {Neal, Radford M},
  booktitle = {Handbook of {M}arkov Chain {M}onte {C}arlo},
  editor    = {Brooks, Steve and Gelman, Andrew and Jones, Galin L
               and Meng, Xiao-Li},
  chapter   = {5},
  pages     = {113--162},
  publisher = {Chapman \& Hall/{CRC}},
  address   = {Boca Raton, FL},
  year      = {2011}
}

@article{battelli2009role,
  title     = {The role of the parietal lobe in visual extinction studied with
               transcranial magnetic stimulation},
  author    = {Battelli, Lorella and Alvarez, George A and Carlson, Thomas
               and Pascual-Leone, Alvaro},
  journal   = {Journal of Cognitive Neuroscience},
  volume    = {21(10)},
  pages     = {1946--1955},
  year      = {2009},
  publisher = {{MIT} Press}
}

@article{libedinsky2011role,
  title     = {Role of prefrontal cortex in conscious visual perception},
  author    = {Libedinsky, Camilo and Livingstone, Margaret},
  journal   = {Journal of Neuroscience},
  volume    = {31(1)},
  pages     = {64--69},
  year      = {2011},
  publisher = {Society for Neuroscience}
}

\newpage 

\appendix
\renewcommand{\thesection}{A.\arabic{section}}
\renewcommand{\thesubsection}{A.\arabic{section}.\arabic{subsection}}
\renewcommand{\thetable}{A.\arabic{table}}
\renewcommand{\thefigure}{A.\arabic{figure}}

\setcounter{table}{0}
\setcounter{figure}{0}

{\hfill \large ANNEXURE \hfill}


\section{Literature Review}
\label{sec:litreview}
In this section of the supplemental material, we highlight previously proposed methods found in literature to handle feature extraction in high dimensional data with complex characteristics. Specifically, we outline techniques introduced to reduce the dimension of the problem and provide a detailed review on shrinkage and selection penalty functions used to induce sparsity into the penalized loss function (PL) estimator. 

Several procedures use a combination of principal component analysis (PCA) and independent component analysis (ICA) or an extension of them to reduce the dimension of the covariate structure. \cite{calhoun2001method} proposed a factor-analytic group decomposition technique applied to functional MRI (fMRI) data, which decomposes the spatial structure using independent components analysis (ICA). \cite{lukic2002spatially} and \cite{svensen2002ica} developed methods using ICA and eigen-decomposition on the covariate structure to map spatial and temporal patterns of active brain regions. \cite{beckman2005tensorial} proposed spatio-temporal decomposition through a tensorial extension of ICA applied to multi-subject fMRI data. \cite{caffo2010twostage} proposed a two-stage method to reduce the dimension of the covariates using subject-level singular-value decomposition (SVD), followed by population-level principal component analysis (PCA). Since there is no clear and concise interpretation of the reduced dimensions obtained from PCA and ICA, the goal of conducting variable selection on associated brain regions is not achieved. 

To highlight the spatio-temporal correlations ever-present in neuroimaging measurements, several methods were proposed involving a combination of roughness and sparse penalties imposed on the spatial and temporal domains. \cite{wanmei2009distributed} used a one-way $\ell_1\ell_2$-regularization method and \cite{huang2009analysis} used two-way regularized SVD methods imposing roughness penalties on both the spatial and temporal domains  \cite{witten2009penalized} and \cite{lee2010biclustering} imposed sparse penalties on both domains, \cite{tian2012twoway} placed a rough penalty and a sparse penalty separately on the two domains. Other techniques impose regularization penalties on a family of tensor GLMs, imposing a low rank approximation on the predictor tensor for dimension reduction \citep{zhou2013tensor, zhou2014regularized}. However, such methods are computationally unattractive and lack effective statistical inference. 

\underline{Sparsity inducing penalty functions}. Although the $\ell_0$-norm penalty induces sparsity, quantifying $\hat{\boldsymbol{\beta}}_{PL}$ remains computationally cumbersome caused by repercussions from the jumpy (discontinuous) nature of indicator functions. As a way of removing the discontinuities established by the indicator function in the $\ell_0$-norm penalty, \cite{tibshirani1996regression} makes use of the $\ell_1$-norm penalty to propose a new convex optimization technique - \textit{LASSO} (least absolute shrinkage and selection operator), $\text{LASSO}(\boldsymbol{\beta} | \lambda) = \lambda \sum_{j=1}^p |\beta_j|,$ where $\lambda \geq 0$ is the tuning parameter. However, there are limitations surrounding variable selection when making use of the LASSO penalty. \cite{zou2005regularization} pointed out that the convex nature of the LASSO problem, when $p \gg n$, allows for no more than n variables to be selected, and further, struggles selecting more than one variable from a group with high pairwise correlations. In response to the preceding restrictions faced with LASSO, \cite{zou2005regularization} developed an elastic net penalty to induce sparsity, $\text{e-net}(\boldsymbol{\beta} | \lambda_1, \lambda_2) = \lambda_1 \sum_{j=1}^p |\beta_j| + \lambda_2 \sum_{j=1}^p \beta_j^2,$ where $\lambda_1, \lambda_2 \geq 0$ are tuning parameters for the LASSO component and ridge component, respectively. The elastic net penalty is a mixture of the LASSO penalty and the ridge penalty introduced to inherit the positive characteristics of LASSO while addressing its limitations. \cite{zou2005regularization} mentions that the naiveness of the elastic net penalty leads to biased estimation and classification when it does not closely resemble the LASSO or ridge penalty.  It was additionally shown by \cite{zou2006adaptive} that employing a common regularization parameter $\lambda$ leads to inconsistent estimates of $\hat{\boldsymbol{\beta}}_{PL}$. To overcome this obstacle, \cite{zou2006adaptive} developed an alternative to LASSO by introducing a weighted $\ell_1$-norm penalty, called adaptive LASSO (a-LASSO), $\text{a-LASSO}(\boldsymbol{\beta} | \lambda, \boldsymbol{w}) = \lambda \sum_{j=1}^p w_j |\beta_j|,$ where $\boldsymbol{w} = (w_1, \dots, w_p)^T$ is the weight vector and $\lambda \geq 0$ is the regularization parameter. With an appropriate choice of weights based on $\sqrt{n}$-consistent estimators for components of $\hat{\boldsymbol{\beta}}_{PL}$, it can be shown that the a-LASSO penalty removes the bias issue by dropping insignificant coefficients \citep{zou2006adaptive}. Selecting suitable weights to achieve the formerly stated advantage of a-LASSO is computationally infeasible as we are not aware of the insignificant predictors beforehand. Furthermore, using the data to construct weights is unrealistic in the presence of high dimensions. \cite{kim2008smoothly} remedy the constraints attached to the a-LASSO method with an iterative concave convex procedure coupled with the smoothly clipped absolute deviation (SCAD) penalty suggested by \cite{fan2001variable}
\begin{align*}
    \text{SCAD}(\boldsymbol{\beta} | \lambda, \gamma) = \lambda \sum_{j=1}^p |\beta_j| \mathbbm{1}_{(|\beta_j| < \lambda)} &+ \sum_{j=1}^p \Big(\frac{\gamma \lambda (|\beta_j| - \lambda) - (|\beta_j|^2 - \lambda^2)^2/2}{\gamma - 1} + \lambda^2\Big) \mathbbm{1}_{(\lambda \leq |\beta_j| \leq \gamma \lambda)} \nonumber \\
    &+ \sum_{j=1}^p \Big(\frac{(\gamma - 1) \lambda^2}{2} + \lambda^2\Big)\mathbbm{1}_{(|\beta_j| \geq \gamma \lambda)}
\end{align*}
where $\lambda \geq 0$ and $\gamma > 1$ is a shrinkage parameter. The drawback of the SCAD penalty used for this purpose is the level of computational and analytical difficulty in finding $\hat{\boldsymbol{\beta}}_{PL}$ coming from the nonconvexity of the optimization problem \citep{zhang2010nearly}. As an alternative, \cite{zhang2010nearly} developed a nearly unbiased estimation and selection method by introducing the minimax concave penalty (MCP)
\begin{align}
    \text{MCP}(\boldsymbol{\beta} | \lambda, \gamma) = \lambda \sum_{j=1}^p \Big[\int_0^{|\beta_j|} (1 - x/(\gamma \lambda))_+dx\Big], \label{mcp}
\end{align}
where $\lambda \geq 0$,s $\gamma > 0$, and subscript $+$ denotes the positive portion of a function.

\section{Details on GD Prior and GD Posterior}
\label{sec:gdpriorpost}

In this section, we elaborate on the details of the GD prior and posterior distribution. As mentioned in the main paper, the GD prior is constructed in a way that the GD estimator acts as a Bayesian version of the PL estimator. In particular, the GD estimator can be shown to be proportional to the PL estimator with the class of BD measures as the choice of loss function and an $\ell_0$-norm approximation as the penalty function. 

\subsection{ Bregman Divergence Measures as Loss Functions}
\label{sec:bdlossfn}

Here, we further discuss the wide-ranging application and flexibility of the class of BD measures, the duality property, and the bijective relationship between the class of BD measures and natural exponential family. Let $\psi: \Omega \rightarrow \mathbbm{R}$ denote a strictly convex and differentiable function on a nonempty convex set $\Omega \subseteq \mathbbm{R}^n$, and let $\boldsymbol{z}_1, \boldsymbol{z}_2 \in \Omega$. The Bregman divergence for $\boldsymbol{z}_1$ and $\boldsymbol{z}_2$ with respect to $\psi$ is defined to be $BD_{\psi}(\boldsymbol{z}_1, \boldsymbol{z}_2) = \psi(\boldsymbol{z}_1) - \psi(\boldsymbol{z}_2) - (\boldsymbol{z}_1 - \boldsymbol{z}_2)^T \nabla \psi(\boldsymbol{z}_2)$, where $\nabla \psi$ represents the gradient vector of $\psi$ \citep{bregman1967relaxation}.  

\underline{Duality property}. Originally discussed in \cite{bernardo2009bayesian}, the duality property affirms that the negative log-likelihood function can be viewed as a loss function. Extending the dual relationship between the likelihood function and loss function to the data-generating distribution, we assume that the negative log-likelihood function is a BD measure. Consider the standard normal likelihood function $f(\boldsymbol{y} | \boldsymbol{\beta}) \propto \exp\{-\frac{1}{2} L_2(\boldsymbol{y}, \boldsymbol{h}(\boldsymbol{X\beta}))\}$, where $L_2(\boldsymbol{z}_1, \boldsymbol{z}_2) = ||\boldsymbol{z}_1 - \boldsymbol{z}_2||_2^2$ denotes the squared Euclidean loss. Hence, $-\log f(\boldsymbol{y} | \boldsymbol{\beta}) \propto L_2(\boldsymbol{y}, \boldsymbol{h}(\boldsymbol{X\beta}))$, i.e., the negative log-likelihood function for the standard normal distribution is proportional to the squared Euclidean loss function. Table \ref{tab:bdloss} below shows some examples of Bregman divergence measures as loss functions for specified strictly convex functions.

\begin{table}[!t]
    \centering
    \caption{Bregman divergence measures as known loss functions.}
    \label{tab:bdloss}
    \begin{tabular}{ccr}
    \hline
    $\psi(\boldsymbol{z})$ & $BD_{\psi}(\boldsymbol{z}_1, \boldsymbol{z}_2)$ & Loss Function \\
\hline
$||\boldsymbol{z}||_2^2$ & $||\boldsymbol{z}_1 - \boldsymbol{z}_2||^2$ & Squared Euclidean loss \\
$\boldsymbol{z}^T \boldsymbol{Wz}$ & $(\boldsymbol{z}_1 - \boldsymbol{z}_2)^T \boldsymbol{W} (\boldsymbol{z}_1 - \boldsymbol{z}_2)$ & Mahalanobis distance \\
$\sum_{i=1}^n z_i\log z_i$ & $\sum_{i=1}^n z_{1i} \log\left(\frac{z_{1i}}{z_{2i}}\right) - (z_{1i} - z_{2i})$ & Kullback-Leibler divergence \\
$\sum_{i=1}^n -logz_i$ & $\sum_{i=1}^n \left\{\frac{z_{1i}}{z_{2i}} - log\left(\frac{z_{1i}}{z_{2i}}\right) - 1\right\}$ & Itakura-Saito distance \\
$\sum_{i=1}^n e^{cz_i}$ & $\sum_{i = 1}^n \left\{e^{cz_{2i}} \left(e^{c(z_{1i} - z_{2i})} - c(z_{1i} - z_{2i}) - 1\right)\right\}$ & Weighted Linex loss \\ \hline
    \end{tabular}
\end{table}

\underline{Bijection with natural exponential family}. An advantage of using the class of BD measures as our choice of loss function is the flexibility they bring into the approach as a result of the bijective relationship existing with the natural exponential family. For instance, if we consider the strictly convex function defined by $\psi(z) = \sum_{i=1}^n z_i\log z_i$, then the BD measure can be written as proportional to the Poisson likelihood function using as follows:
\begin{align*}
    f(\boldsymbol{y} | \boldsymbol{\beta}) &\propto \exp\left\{ - \sum_{i=1}^n \left[y_i \log \left(\frac{y_i}{h(\boldsymbol{x}_i^T\boldsymbol{\beta})}\right) - (y_i - h(\boldsymbol{x}_i^T\boldsymbol{\beta})\right]\right\} \\
    &\propto \prod_{i=1}^n \left(e^{-h(\boldsymbol{x}_i^T\boldsymbol{\beta})} \left[h(\boldsymbol{x}_i^T\boldsymbol{\beta})\right]^{y_i}\right).
\end{align*}
Table \ref{tab:bdexamples} provides examples of the bijective relationship between BD measures and the natural exponential family. Using the BD measure corresponding to the Bernoulli distribution, we have 
\begin{align*}
    BD_{\psi}(\boldsymbol{y}, h(\boldsymbol{X} \boldsymbol{\beta})) &= \sum_{i = 1}^N \left\{y_i \log\left(\frac{y_i}{h(\boldsymbol{x}_i^T \boldsymbol{\beta})}\right) + (1 - y_i)\log\left(\frac{1 - y_i}{1 - h(\boldsymbol{x}_i^T \boldsymbol{\beta})}\right)\right\} \\
    &= \sum_{i = 1}^N \left\{ \log\left(\left[\frac{y_i}{h(\boldsymbol{x}_i^T \boldsymbol{\beta})}\right]^{y_i} \left[\frac{1 - y_i}{1 - h(\boldsymbol{x}_i^T \boldsymbol{\beta})}\right]^{1 - y_i}\right)\right\},
\end{align*}
and the likelihood function can be written as 
\allowdisplaybreaks
\begin{equation*}
    f_{\psi}(\boldsymbol{y} | \boldsymbol{\beta}) \propto \exp\{-BD_{\psi}(\boldsymbol{y}, h(\boldsymbol{X} \boldsymbol{\beta}))\} \propto \prod_{i = 1}^N \left( \left[h(\boldsymbol{x}_i^T \boldsymbol{\beta})\right]^{y_i} \left[1 - h(\boldsymbol{x}_i^T \boldsymbol{\beta})\right]^{1 - y_i} \right),
\end{equation*}
which is the Bernoulli likelihood function. It is worth mentioning that a more general distribution family, relative to the natural exponential family, can be generated by Bregman divergence measures. For instance, \cite{zhang2009new} demonstrated that the quasi-likelihood function \citep{wedderburn1974quasi} can be generated by BD measures.

\begin{table}[!t]
    \centering
    \caption{Bregman divergence measures and their corresponding distributions in the natural exponential family along with the associated strictly convex function.}
    \label{tab:bdexamples}
    \begin{tabular}{ccr}
       \hline
$\psi(z)$ & $BD_{\psi}(z_1, z_2)$ & Distribution \\
\hline
$\frac{1}{2\sigma^2}z^2$ & $\frac{1}{2\sigma^2}(z_1 - z_2)^2$ & Gaussian \\
$z\log z$ & $z_1 \log\left(\frac{z_1}{z_2}\right) - (z_1 - z_2)$ & Poisson \\
$-\log z$ & $\frac{z_1}{z_2} - \log\left(\frac{z_1}{z_2}\right) - 1$ & Exponential \\
$z\log z + (1 - z)\log (1 - z)$ & $z_1\log\left(\frac{z_1}{z_2}\right) + (1 - z_1)\log\left(\frac{1 - z_1}{1 - z_2}\right)$ & Bernoulli \\
\hline
    \end{tabular}
\end{table}

\subsection{ GD Penalty}
\label{sec:gdpenalty}

Here, we emphasize the advantage of the GD penalty over the SELO penalty. In particular, we demonstrate numerically that the GD penalty is differentiable everywhere and that the SELO penalty is not differentiable at $\beta_j = 0$. \cite{dicker2013variable} defined the SELO penalty as $\text{SELO}(\boldsymbol{\beta} | \lambda, \tau_0) = \frac{\lambda}{\log (2)} \sum_{j=1}^p \log(\frac{|\beta_j|}{|\beta_j| + \tau_0} + 1)$, where $\lambda \geq 0$ is a tuning parameter and $\tau_0 > 0$ is a deterministic constant. Note, as $\tau_0 \rightarrow 0$, $\log (\frac{|\beta_j|}{|\beta_j| + \tau_0} + 1) \rightarrow \log (2)$ for $|\beta_j| \neq 0$. Hence, the SELO penalty converges to the $\ell_0$-norm penalty as $\tau_0$ tends to 0.

\underline{Advantage of GD penalty over SELO penalty}. Although the SELO penalty is also a continuous approximation of the $\ell_0$-norm penalty, it has restrictions. At $\beta_j = 0$, the SELO penalty is not differentiable. Whereas, the GD penalty is differentiable everywhere. Due to the smoothness of the GD penalty function, applying the Newton-Rhapson algorithm is stable and convergence is fast. We provide a numerical demonstration of this advantage the GD penalty has over the SELO penalty. Observe, 
\begin{equation*}
    \frac{\delta \text{SELO}(\boldsymbol{\beta} | \lambda, \tau_0)}{\delta \beta_j} = \frac{\lambda}{log(2)} \times \frac{\tau_0 \beta_j}{|\beta_j|(2|\beta_j| + \tau_0)(|\beta_j| + \tau_0)^2},
\end{equation*}
which does not exist at $\beta_j = 0$. On the other hand,
\begin{equation*}
    \frac{\delta \tilde{\ell}_0(\boldsymbol{\beta} | \lambda, \tau)}{\delta \beta_j} = \lambda \times \frac{2\{X_{[, j]}^T X_{[, j]}\}\beta_j^2\tau_0^2}{(\tau_0^2 + \{X_{[, j]}^T X_{[, j]}\}\beta_j^2)^2},
\end{equation*}
which is 0 at $\beta_j = 0$. This shows that the GD penalty is continuous and differentiable everywhere, and the SELO penalty is continuous, but not differentiable everywhere.

\subsection{ GD Posterior Distribution}
\label{sec:gdpostdist}

It is known that if the prior distribution is proper and the likelihood function is non-degenerate, then the set of points leading to an improper posterior distribution has Lebesgue measure zero, i.e., the posterior distribution is proper almost everywhere. Further, if the likelihood function is uniformly bounded, then the posterior distribution is proper everywhere. Notice, for all $\boldsymbol{\beta} \in \mathbbm{R}^L$, the GD prior, $\pi_G(\boldsymbol{\beta} | \hat{\boldsymbol{d}}_{MAP})$, is proper, and the likelihood function, $f_{\psi}(\boldsymbol{y} | \boldsymbol{\beta})$, is non-degenerate and bounded. One can see that the GD prior is proper from Lemma \ref{l0converge} and by observing the form of $\tilde{\ell}_0(\cdot | \lambda , \tau)$. Also, to verify that the likelihood function is non-degenerate and bounded, one can recall the form of the likelihood, which uses the duality property, and note that BD measures are non-negative. Thus, the GD posterior distribution is proper. However, posterior computation in high dimensional problems requires more careful attention.

We have several mechanisms to summarize the information obtained from the MCMC samples. We may compute the posterior mode (MAP estimate), posterior mean, or posterior median, and conduct inference based on point estimates. The posterior distribution of $\boldsymbol{\beta}$ can be used to assess the uncertainty associated with our estimation procedure. Thus, we are able to calculate credible intervals and highest posterior density (HPD) intervals for the model parameters. The posterior distribution can be used to generate posterior probabilities to determine the ``best" model or set of models, and generate further probabilities of interest in variable selection \citep{mitchell1988bayesian}. We may also use Bayesian model averaging \citep{hoeting1999bayesian} to incorporate multiple scenarios explored by the MCMC sampler.


\section{Simulation Settings and Details}
\label{sec:simsettings}
In longitudinal simulation studies, only a single time point is considered, so one can simply generate the data, $\boldsymbol{X}_t$, which can be used along with the true coefficient vector, $\boldsymbol{\beta}_t$, and known link function, $h(\cdot)$, to generate the response vector $\boldsymbol{y}$. A major challenge in simulating data for our framework is that $\boldsymbol{X}_t$ and $\boldsymbol{\beta}_t$ vary over time, while the response vector remains the same, i.e., $\boldsymbol{y}_t = \boldsymbol{y}$, for all $t = 1, \dots, \tau$. To overcome this hurdle, we generate a Bernoulli response vector $\boldsymbol{y} \in \mathbbm{R}^{n}$ with a fixed probability, we set $P(y_i = 1) = 77/122$ to imitate the real data. The three-dimensional data array $\chi \in \mathbbm{R}^{n \times L \times \tau}$ is constructed by first generating $\mathcal{U} \in \mathbbm{R}^{N \times L \times \tau}$ from a multivariate normal distribution with mean $\boldsymbol{0}$, such that $N$ is chosen to be sufficiently larger than $n$. Taking into account the high subject-to-subject variability in the data, the covariance matrix is constructed by $\Sigma_{ij} = \rho^{|i - j|}$, $i,j = 1, \dots, L$, where $0 < \rho < 1$ is chosen to incorporate correlations between locations. In our simulations, we choose $\rho = 0.5$. For $t = 1, \dots, \tau$, let $\boldsymbol{X}_t \in \mathbbm{R}^{n \times L}$ denote the matrix along the third dimension of $\chi$ and $\boldsymbol{U}_t \in \mathbbm{R}^{N \times L}$ denote the matrix along the third dimension of $\mathcal{U}$. We generate rows of $\boldsymbol{X}_t$ from the larger collection of rows of $\boldsymbol{U}_t$ so that $\boldsymbol{X}_t \boldsymbol{\beta}_t$ corresponds to $\boldsymbol{y}$, for all $t$. See Figure \ref{fig:localmods} for a demonstration of our local modeling approach. 

\begin{figure}[!t]
\center
    \makebox[\textwidth][c]{\includegraphics[scale=0.65, trim=1.5cm 4cm 4.25cm 7.25cm, clip]{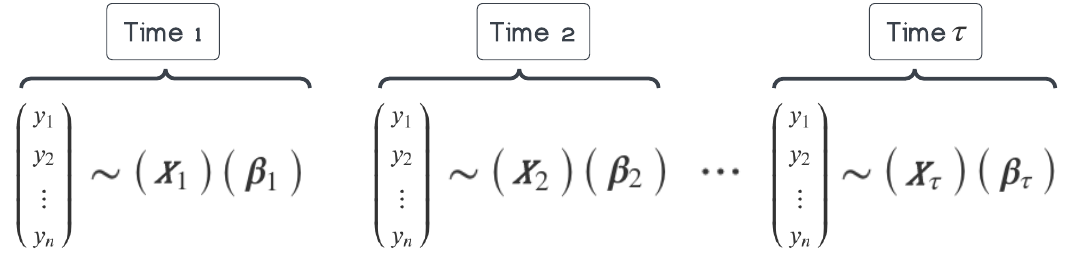}}
    \caption{Illustration of local Bayesian models.}
    \label{fig:localmods}
\end{figure}

We generate the true coefficient matrix $\boldsymbol{\beta} = (\beta_{lt})_{l = 1, \dots, L, t = 1, \dots, \tau}$ by sampling entries from a Bernoulli distribution and randomly assigning $20\%$ of the locations as active. To differentiate between active and inactive locations, entries in active rows are assigned lower probability of being 0 compared to inactive locations. Entries in inactive locations are 0 with probability 0.8 and active locations are 0 with probability 0.2. In other words, active locations should have more 1 entries, whereas inactive locations should have more 0 entries. We then randomly select $30\%$ of the coefficients to be multiplied by -1 to account for negative electrical signals. After we generate the coefficient matrix, we fix the response vector and generate the data array as described above.

Due to the low signal-to-noise ratio observed in EEG data, caused by the non-invasiveness of the procedure, we add random noise to the data from a normal distribution centered at 0. By adding random noise to the data, we also account for the sporadic behavior seen in EEG measurements. Without doing this, the simulation settings are infeasible and the results will not accurately depict the performance of the model when applied to real data. We explored three different values of $noiseSD$, 1, 1.5, and 2. Table \ref{tab:simres1-noise1} and Table \ref{tab:simres2-noise1} contain simulation results from using a standard deviation of 1, i.e., $\text{SNR} = 1$. Table \ref{tab:simres1-noise2} and Table \ref{tab:simres2-noise2} contain simulation results from using a standard deviation of 2, i.e., $\text{SNR} = 1/2$.

\begin{table}[!t]
\centering
\caption{Simulation results: estimation and feature extraction; $\text{SNR} = 1$.} 
\begin{tabular}{ccccc}
\hline 
$L$ & $\tau$ & $rMSE_{est}$ & corrInd & incorrInd \\
\hline
25 & 100 & 0.516 (0.010) & 0.946 (0.094) & 0.035 (0.047) \\
& 175 & 0.514 (0.007) & 0.980 (0.060) & 0.008 (0.012) \\
& 250 & 0.514 (0.006) & 0.990 (0.044) & 0.004 (0.017) \\
\hline
60 & 100 & 0.573 (0.045) & 0.793 (0.131) & 0.184 (0.094) \\
& 175 & 0.565 (0.023) & 0.860 (0.107) & 0.140 (0.078) \\
& 250 & 0.566 (0.025) & 0.903 (0.093) & 0.117 (0.085) \\
\hline
75 & 100 & 0.826 (0.101) & 0.714 (0.125) & 0.246 (0.103) \\
& 175 & 0.825 (0.088) & 0.789 (0.126) & 0.242 (0.094) \\
& 250 & 0.841 (0.070) & 0.841 (0.108) & 0.226 (0.097) \\
\hline
\end{tabular}
\label{tab:simres1-noise1}
\end{table}

\begin{table}[!t]
\centering
\caption{Simulation results: prediction; $\text{SNR} = 1$.}
\begin{tabular}{cccccc}
\hline
$L$ & $\tau$ & TPR & FPR & PE & AUC \\
\hline
25 & 100 & 0.990 (0.022) & 0.016 (0.038) & 0.012 (0.019) & 0.999 (0.003) \\
& 175 & 0.997 (0.013) & 0.002 (0.015) & 0.003 (0.010) & 1.000 (0.000) \\
& 250 & 1.000 (0.000) & 0.001 (0.007) & 0.000 (0.003) & 1.000 (0.000) \\
\hline
60 & 100 & 0.970 (0.048) & 0.034 (0.059) & 0.031 (0.036) & 0.995 (0.011) \\
& 175 & 0.991 (0.026) & 0.005 (0.022) & 0.007 (0.017) & 0.999 (0.003) \\
& 250 & 0.998 (0.012) & 0.000 (0.000) & 0.001 (0.006) & 1.000 (0.001) \\
\hline
75 & 100 & 0.967 (0.040) & 0.033 (0.057) & 0.032 (0.034) & 0.994 (0.011) \\
& 175 & 0.995 (0.018) & 0.003 (0.018) & 0.005 (0.013) & 1.000 (0.002) \\
& 250 & 1.000 (0.005) & 0.003 (0.020) & 0.001 (0.007) & 1.000 (0.001) \\
\hline
\end{tabular}
\label{tab:simres2-noise1}
\end{table}

For $\text{SNR} = 1$, we can see that the model performs extremely well. In particular, when $\tau$ is large, the model achieves near perfect prediction accuracy. Although we observe a jump in rMSE when $L$ changes from 60 to 75 in the $\text{SNR} = 1$ case, we observe a distinctively larger jump in $\text{SNR} = 1/2$ case. It is reasonable to see this behavior due to the fact that lower $\text{SNR}$ will have a greater impact when the number of locations is large.

\begin{table}[!t]
\centering
\caption{Simulation results: estimation and feature extraction; $\text{SNR} = 1/2$.} 
\begin{tabular}{ccccc}
\hline 
$L$ & $\tau$ & $rMSE_{est}$ & corrInd & incorrInd \\
\hline
25 & 100 & 0.553 (0.009) & 0.818 (0.166) & 0.174 (0.124) \\
& 175 & 0.552 (0.006) & 0.876 (0.153) & 0.144 (0.129) \\
& 250 & 0.552 (0.005) & 0.876 (0.147) & 0.096 (0.112) \\
\hline
60 & 100 & 0.581 (0.040) & 0.613 (0.168) & 0.305 (0.089) \\
& 175 & 0.582 (0.030) & 0.637 (0.181) & 0.285 (0.099) \\
& 250 & 0.583 (0.028) & 0.682 (0.166) & 0.301 (0.094) \\
\hline
75 & 100 & 0.966 (0.142) & 0.611 (0.138) & 0.334 (0.083) \\
& 175 & 0.971 (0.092) & 0.691 (0.139) & 0.322 (0.098) \\
& 250 & 0.979 (0.072) & 0.735 (0.131) & 0.342 (0.104) \\
\hline
\end{tabular}
\label{tab:simres1-noise2}
\end{table}

\begin{table}[!t]
\centering
\caption{Simulation results: prediction; $\text{SNR} = 1/2$.}
\begin{tabular}{cccccc}
\hline
$L$ & $\tau$ & TPR & FPR & PE & AUC \\
\hline
25 & 100 & 0.787 (0.107) & 0.215 (0.136) & 0.211 (0.093) & 0.872 (0.090) \\
& 175 & 0.875 (0.078) & 0.116 (0.099) & 0.122 (0.063) & 0.950 (0.043) \\
& 250 & 0.902 (0.074) & 0.105 (0.098) & 0.100 (0.066) & 0.966 (0.035) \\
\hline
60 & 100 & 0.793 (0.116) & 0.220 (0.127) & 0.213 (0.090) & 0.865 (0.088) \\
& 175 & 0.832 (0.114) & 0.153 (0.112) & 0.162 (0.089) & 0.919 (0.072) \\
& 250 & 0.895 (0.079) & 0.120 (0.100) & 0.110 (0.070) & 0.957 (0.052) \\
\hline
75 & 100 & 0.884 (0.083) & 0.113 (0.094) & 0.114 (0.066) & 0.955 (0.043) \\
& 175 & 0.935 (0.058) & 0.055 (0.071) & 0.061 (0.048) & 0.986 (0.023) \\
& 250 & 0.979 (0.036) & 0.026 (0.053) & 0.023 (0.032) & 0.996 (0.012) \\
\hline
\end{tabular}
\label{tab:simres2-noise2}
\end{table}

The overall trends observed in estimation, feature extraction, and prediction are similar to those observed in Table 3 and 4 when $\text{SNR} = 1/1.5$. As expected, the overall performance and accuracy of the model dwindle as the $\text{SNR}$ decreases. With a low $\text{SNR}$, it is challenging for the model to distinguish the meaningful signal from the noise, which leads to decreased accuracy in classifying subjects and increased difficulty in identifying active locations. Also, low $\text{SNR}$ can lead to increased sensitivity to minor fluctuations or variations in the data.

\paragraph{Sensitivity of the Stage~1 FDR Threshold $c$.} To assess whether Stage~1 variable selection is sensitive to the FDR constant $c$ used to compute posterior inclusion probabilities, we repeated the representative simulation setting (SNR = 1/2) under three choices, $c \in \{10^{-4}, 10^{-3}, 10^{-2}\}$. The results show robustness across this range. In particular, the true-positive rate is identical for $c=10^{-4}$ and $c=10^{-3}$ ($\mathrm{TPR}=0.64$ in both cases), while the false-positive rate remains zero and AUC is $0.944$ throughout all three choices of $c$. The only noticeable difference is a small reduction in $\mathrm{TPR}$ at the largest threshold considered ($\mathrm{TPR}=0.60$ at $c=10^{-2}$), without any increase in false positives or loss in AUC. Overall, these findings indicate that the Stage~1 screening and downstream prediction are stable over two orders of magnitude around our chosen value. Selecting $c=10^{-3}$ therefore preserves the best observed recovery among the tested values while maintaining a deliberately conservative stance (zero false positives), which is appropriate given that the true signal-to-noise ratio in the EEG application is unknown and may be lower than in the simulation setting.

\section{EEG Analysis and Visualization}
\label{sec:eeganalysis}

In this section, we describe the graphical tools used to visualize the EEG data, expand on the programming approach used to fit the EEG data, and provide additional interpretations of the results.

\underline{ERP plots and scalp maps}. Although more advanced approaches to signal processing of biophysical signals in the brain are available, the use of amplitude and latency measures of peaks in EEG trial averages, known as event related potentials (ERPs), are still heavily utilized in electrophysiology research \citep{delorme2004eeglab}. An ERP-image plot contains several time series curves (one for each location) which represent the average ``potential" (amongst all subjects) at each latency. \cite{delorme2004eeglab} additionally described a data visualization method referred to as scalp maps, which show the topographical behavior of the distribution of average potential at a specified time point.

\underline{Fitting EEG data}. Prior to using LOOCV, we were averaging the results over 100 replications as we did in the simulations. That is, we randomly generated 100 different training and testing sets, fit the model on each of the training sets, analyzed the results using the corresponding testing sets, and averaged the results from each of the 100 replications. It was clear that certain seeds resulted in poor splits of the data. We could see that some splits returned high tpr's and low fpr's, while others had values of tpr's and fpr's close to 0.5. It followed that the AUC values below 0.5 corresponded to the same bad splits. To avoid ``bad" splits, we decided to use LOOCV. The results from LOOCV were much nicer and proved that the bad splits were the issue before. Out of all 122 folds, an average of approximately $3.87$ locations were selected, which accounts for approximately $6.79\%$ of the 57 locations being considered in the data set. We tried using a union of all of the locations selected at least once, but this only amounted to 18/57 ($\approx 31.58\%$) of the locations. We found that majority of the folds resulted in the model selecting $3$ of the $57$ locations as active. To be specific, $74/122$ ($\approx 60.66\%$) of the folds selected exactly $3$ locations as active, and the indices of the selected locations were exactly the same: 35 (CP5), 43 (P7), 52 (P07). Interesting enough, all 122 of the folds selected electrodes P7 and P07, and 121 folds selected CP5; that is, CP5, P7, and P07 were all selected in 121/122 folds. Aside from those three, location 2 (FPZ) was selected in 39/122 ($\approx 31.97\%$), location 16 (FT7) was selected in 20/122 ($\approx 16.39\%$), and location 44 (P5) was selected in 15/122 ($\approx 12.30\%$). For reference to the scalp map pertaining to the frequency of selection of locations, we provide a topographic image of the scalp in Figure \ref{fig:brainElecPlot} which specifies the precise locations of the electrodes. 
\begin{figure}[!ht]
    \centering
    \includegraphics[scale=0.53, trim=0.75cm 0.75cm 0.75cm 0.75cm, clip]{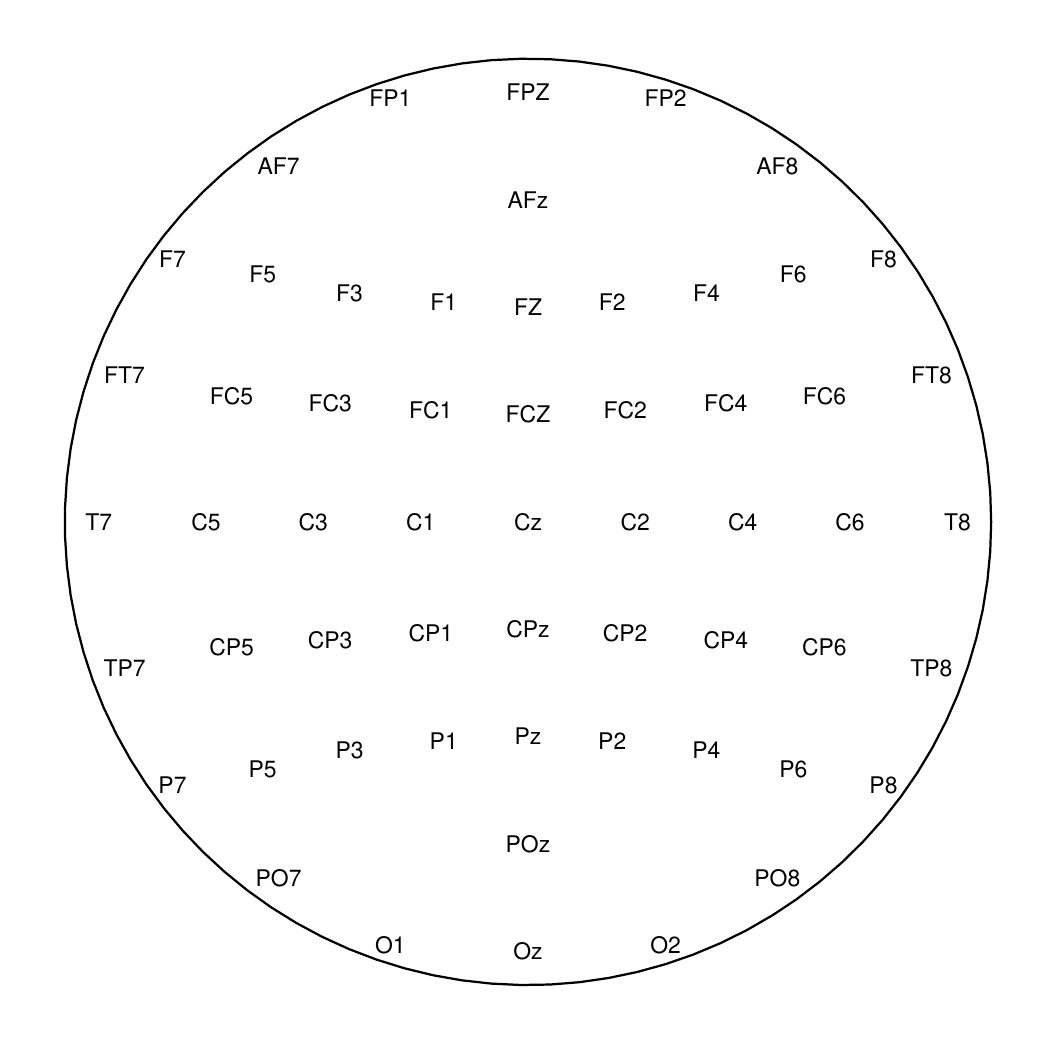}
    \caption{Showing names of electrodes at their corresponding locations of the brain.}
    \label{fig:brainElecPlot}
\end{figure}


\section{Proofs of Useful Results}
\label{sec:proofs}

\begin{lem}
    If $g(\tau_0 | x) = \frac{x^2}{\tau_0^2 + x^2}$, then $\underset{\tau_0 \rightarrow 0}{lim} \text{ } g(\tau_0 | x = 0) = 0$ and $\underset{\tau_0 \rightarrow 0}{lim} \text{ } g(\tau_0 | x \neq 0) = 1.$
    \label{l0converge}
\end{lem}
The proof for Lemma \ref{l0converge} is omitted as the result is elementary.

\begin{lem}
    Let $\pi_{GD}(\boldsymbol{\beta}, \boldsymbol{d})$ denote the GD prior. Then, for any $\lambda \geq 0$ and $\tau_0 > 0$, $\underset{\boldsymbol{d}}{\text{arg max}}\{\pi_{GD}(\boldsymbol{\beta}, \boldsymbol{d})\} = (\frac{2\lambda}{\tau_1^2 + \beta_1^2}, \dots, \frac{2\lambda}{\tau_p^2 + \beta_p^2})$
\begin{align*}
    \underset{\boldsymbol{d}}{\text{arg max}}\{\pi_{GD}(\boldsymbol{\beta}, \boldsymbol{d})\} = \left(\frac{2\lambda}{\tau_1^2 + \beta_1^2}, \dots, \frac{2\lambda}{\tau_p^2 + \beta_p^2}\right),
\end{align*}
where $\tau_j^2 = \tau_0^2/\{(\boldsymbol{X}_{[,j]})^T\boldsymbol{X}_{[,j]}\}$, $j = 1, \dots, p$. Further, if $\hat{\boldsymbol{d}} = \underset{\boldsymbol{d}}{\text{argmax}}\{\pi_{GD}(\boldsymbol{\beta}, \boldsymbol{d})\}$, then
\setlength{\belowdisplayskip}{0pt}
\begin{align*}
    \pi_{GD}(\boldsymbol{\beta}, \hat{\boldsymbol{d}}) &\propto \exp\left\{- \tilde{\ell_0}(\boldsymbol{\beta} | \lambda, \tau_0)\right\}.
\end{align*}
\label{dmap}
\end{lem}

\noindent{\textbf{Proof of Lemma 2:}} We will find the arg max of the GD prior $\pi_{GD}(\boldsymbol{\beta}, \boldsymbol{d})$ with respect to a single $d_j$ for simplicity, as the form of $\hat{d}_j$ will be similar for each $j$. Recall, 
\begin{gather*}
    \pi_{GD}(\boldsymbol{\beta}, \boldsymbol{d}) \propto \pi_G(\boldsymbol{\beta} | \boldsymbol{d}) \pi_D(\boldsymbol{d}),
\end{gather*}
where $\pi_G(\boldsymbol{\beta} | \boldsymbol{d})$ and $\pi_D(\boldsymbol{d})$. Note, the arg max of a function and its natural log are equivalent. We will work with the log of the GD prior for easier derivation. Observe, 
\begin{gather*}
    \frac{\partial log\pi_{GD}(\boldsymbol{\beta}, \boldsymbol{d})}{\partial d_j} = \frac{1}{2d_j} - \frac{\beta_j^2}{2} + \frac{2\lambda - 1}{2d_j} - \frac{\tau_0^2}{2\{(\boldsymbol{X}_{[,j]})^T\boldsymbol{X}_{[,j]}\}}. 
\end{gather*}
Then, 
\begin{align*}
    \frac{\partial log\pi_{GD}(\boldsymbol{\beta}, \boldsymbol{d})}{\partial d_j} \Big|_{d_j = \hat{d}_j} = 0 &\implies 2\lambda = \hat{d}_j \beta_j^2 + \hat{d}_j \frac{\tau_0^2}{\{(\boldsymbol{X}_{[,j]})^T\boldsymbol{X}_{[,j]}\}} \\
    &\implies 2 \lambda = \hat{d}_j (\beta_j^2 + \frac{\tau_0^2}{\{(\boldsymbol{X}_{[,j]})^T\boldsymbol{X}_{[,j]}\}}) \\
    &\implies \hat{d}_j = \frac{2\lambda}{\tau_j^2 + \beta_j^2},
\end{align*}
where $\tau_j^2 = \frac{\tau_0^2}{\{(\boldsymbol{X}_{[,j]})^T\boldsymbol{X}_{[,j]}\}}$, $j = 1, \dots, p$. Further, 
\begin{align*}
    \frac{\partial^2 log\pi_{GD}(\boldsymbol{\beta}, \boldsymbol{d})}{\partial d_j^2} \Big|_{d_j = \hat{d}_j} &= -\frac{1}{2\hat{d}_j^2} - \frac{2\lambda - 1}{2\hat{d}_j^2} = -\frac{\lambda}{\hat{d}_j^2} \leq 0.
\end{align*}
Therefore, 
\begin{gather*}
    \underset{\boldsymbol{d}}{\text{arg max}}\{\pi_{GD}(\boldsymbol{\beta}, \boldsymbol{d})\} = (\frac{2\lambda}{\tau_1^2 + \beta_1^2}, \dots, \frac{2\lambda}{\tau_p^2 + \beta_p^2}).
\end{gather*}
Now, let $\hat{\boldsymbol{d}} = \underset{\boldsymbol{d}}{\text{arg max}}\{\pi_{GD}(\boldsymbol{\beta}, \boldsymbol{d})\}$, $BD_{\psi}(\cdot, \cdot)$ denote a Bregman divergence measure with $\psi$ denoting a strictly convex function, and $\tilde{\ell_0}(\cdot | \lambda, \tau_0)$ denote the $\ell_0$-norm approximation. Since $\pi_{GD}(\boldsymbol{\beta}, \boldsymbol{d}) \propto \pi_G(\boldsymbol{\beta} | \boldsymbol{d})$ with respect to $\boldsymbol{\beta}$, for $\tau_0$ small enough,  
\setlength{\belowdisplayskip}{3pt}
\begin{align*}
    \pi_{GD}(\boldsymbol{\beta}, \hat{\boldsymbol{d}}) &\propto \prod_{j = 1}^p \{\hat{d}_j^{1/2} exp(-\frac{\hat{d}_j}{2} \beta_j^2)\} \\
    &\propto \prod_{j = 1}^p \{(2\lambda \frac{\{(\boldsymbol{X}_{[,j]})^T\boldsymbol{X}_{[,j]}\} \beta_j^2}{\tau_0^2 + \{(\boldsymbol{X}_{[,j]})^T\boldsymbol{X}_{[,j]}\} \beta_j^2})^{1/2} exp(-\lambda \frac{\{(\boldsymbol{X}_{[,j]})^T\boldsymbol{X}_{[,j]}\} \beta_j^2}{\tau_0^2 + \{(\boldsymbol{X}_{[,j]})^T\boldsymbol{X}_{[,j]}\} \beta_j^2} \beta_j^2)\} \\
    &\propto exp\{- \lambda \sum_{j=1}^p (\frac{\{(\boldsymbol{X}_{[,j]})^T\boldsymbol{X}_{[,j]}\} \beta_j^2}{\tau_0^2 + \{(\boldsymbol{X}_{[,j]})^T\boldsymbol{X}_{[,j]}\} \beta_j^2})\} \\
    &= exp\{- \tilde{\ell_0}(\boldsymbol{\beta} | \lambda, \tau_0)\},
\end{align*}
which gives us the desired result. 

\begin{rem}
    Let $\hat{\boldsymbol{d}} = \underset{\boldsymbol{d}}{\text{argmax}}\left\{\underset{\boldsymbol{\beta}}{\text{max}}f_{\psi}(\boldsymbol{y} | \boldsymbol{\beta}) \pi_{GD}(\boldsymbol{\beta},\boldsymbol{d})\right\}$. Then, for any $\lambda \geq 0$ and $\tau_0 > 0$,
    \begin{equation*}
        \underset{\boldsymbol{\beta}}{\text{arg max}}\left\{f_{\psi}(\boldsymbol{y} | \boldsymbol{\beta}) \pi_{GD}(\boldsymbol{\beta}, \hat{\boldsymbol{d}})\right\} = \underset{\boldsymbol{\beta}}{\text{arg min}}\left\{BD_{\psi}\{\boldsymbol{y}, \boldsymbol{h(X\beta)}\} + \tilde{\ell_0}(\boldsymbol{\beta} | \lambda, \tau_0)\right\}.
    \end{equation*}
\label{argmax}
\end{rem}

\noindent{\textbf{Proof of Remark 1:}} Observe that by applying the dual property between the likelihood function and loss function, we get
\begin{align*}
    f_{\psi}(\boldsymbol{y} | \boldsymbol{\beta}) \pi_{GD}(\boldsymbol{\beta}, \hat{\boldsymbol{d}}) &\propto f_{\psi}(\boldsymbol{y} | \boldsymbol{\beta}) exp\{-\stackrel{\sim}{\ell_0} (\boldsymbol{\beta} | \lambda, \tau_0)\} \\
    &\propto exp\{-BD_{\psi}\{\boldsymbol{y}, \boldsymbol{h(X\beta)}\}\} exp\{-\stackrel{\sim}{\ell_0} (\boldsymbol{\beta} | \lambda, \tau_0)\} \\
    &= exp\{-(BD_{\psi}\{\boldsymbol{y}, \boldsymbol{h(X\beta)}\} + \stackrel{\sim}{\ell_0} (\boldsymbol{\beta} | \lambda, \tau_0))\}.
\end{align*}
Hence, 
\begin{align*}
    \underset{\boldsymbol{\beta}}{\text{arg max}}\{f_{\psi}(\boldsymbol{y} | \boldsymbol{\beta}) \pi_{GD}(\boldsymbol{\beta}, \hat{\boldsymbol{d}})\} = \underset{\boldsymbol{\beta}}{arg max} [exp\{-(BD_{\psi}\{\boldsymbol{y}, \boldsymbol{h(X\beta)}\} + \stackrel{\sim}{\ell_0} (\boldsymbol{\beta} | \lambda, \tau_0))\}].
\end{align*}
Since the natural log is a monotonic strictly increasing function and the arg max of a function is the same as the arg min of the negative of that function for strictly convex (or concave) functions, 
\begin{align*}
    \underset{\boldsymbol{\beta}}{arg max} [exp\{-(BD_{\psi}\{\boldsymbol{y}, \boldsymbol{h(X\beta)}\} + \stackrel{\sim}{\ell_0} (\boldsymbol{\beta} | \lambda, \tau_0))\}] &= \underset{\boldsymbol{\beta}}{arg max} \{-(BD_{\psi}\{\boldsymbol{y}, \boldsymbol{h(X\beta)}\} + \stackrel{\sim}{\ell_0} (\boldsymbol{\beta} | \lambda, \tau_0))\} \\
    &= \underset{\boldsymbol{\beta}}{arg min} \{BD_{\psi}\{\boldsymbol{y}, \boldsymbol{h(X\beta)}\} + \stackrel{\sim}{\ell_0} (\boldsymbol{\beta} | \lambda, \tau_0) \}.
\end{align*}
Therefore, it follows that
\begin{align*}
    \underset{\boldsymbol{\beta}}{\text{arg max}}\{f_{\psi}(\boldsymbol{y} | \boldsymbol{\beta}) \pi_{GD}(\boldsymbol{\beta}, \hat{\boldsymbol{d}})\} = \underset{\boldsymbol{\beta}}{\text{arg min}}\{BD_{\psi}\{\boldsymbol{y}, \boldsymbol{h(X\beta)}\} + \stackrel{\sim}{\ell_0}(\boldsymbol{\beta} | \lambda, \tau_0)\}.
\end{align*}

\end{document}